\documentclass[acmsmall,screen]{acmart}

\setcopyright{rightsretained}
\acmPrice{}
\acmDOI{10.1145/3371103}
\acmYear{2020}
\copyrightyear{2020}
\acmJournal{PACMPL}
\acmVolume{4}
\acmNumber{POPL}
\acmArticle{35}
\acmMonth{1}

\usepackage{booktabs}   
\usepackage{subcaption} 
\usepackage{xcolor}
\usepackage[inline]{enumitem}
\usepackage{pifont}

\theoremstyle{plain}
\newtheorem{theorem}{Theorem}
\theoremstyle{definition}
\newtheorem{definition}{Definition}
\newtheorem{lemma}[theorem]{Lemma}

\newtheorem{example}{Example}

\newcommand{\ucomment}[1]{}
\newcommand{\adcomment}[1]{}
\newcommand{\mpcomment}[1]{}
\newcommand{\pkcomment}[1]{}

\newcommand{\seclabel}[1]{\label{sec:#1}}
\newcommand{\secref}[1]{Section~\ref{sec:#1}}

\newcommand{\exlabel}[1]{\label{ex:#1}}
\newcommand{\exref}[1]{Example~\ref{ex:#1}}
\newcommand{\deflabel}[1]{\label{def:#1}}
\newcommand{\defref}[1]{Definition~\ref{def:#1}}
\newcommand{\thmlabel}[1]{\label{thm:#1}}
\newcommand{\thmref}[1]{Theorem~\ref{thm:#1}}

\newcommand{\lemlabel}[1]{\label{lem:#1}}

\newcommand{\itemlabel}[1]{\label{itm:#1}}
\newcommand{\itemref}[1]{(\ref{itm:#1})}

\newcommand{\nats}{\mathbb{N}}

\newcommand{\Cc}{\mathcal{C}}

\newcommand{\Ff}{\mathcal{F}}

\newcommand{\Aa}{\mathcal{A}}
\newcommand{\Mm}{\mathcal{M}}

\newcommand{\Ii}{\mathcal{I}}

\newcommand{\Pp}{\mathcal{P}}

\renewcommand{\vec}[1]{{\bf #1}}
\newcommand{\set}[1]{\{#1\}}
\newcommand{\setpred}[2]{\{#1 \, \mid \, #2\}}
\newcommand{\setpredsmall}[2]{\{#1 \mid #2\}}

\newcommand{\angular}[1]{\langle #1 \rangle}

\newcommand{\delequal}{\overset{\Delta}{=}}
\renewcommand{\emptyset}{\varnothing}

\newcommand{\xdownarrow}[1]{%
  {\left\downarrow\vbox to #1{}\right.\kern-\nulldelimiterspace}
}

\newcommand{\stmt}{\angular{stmt}}
\newcommand{\cond}{\angular{cond}}

\newcommand{\code}[1]{\texttt{#1}}
\newcommand{\codekey}[1]{\textbf{#1}}
\newcommand{\cd}[1]{\code{#1}}
\newcommand{\pfalse}{\codekey{false}}
\newcommand{\pskip}{\codekey{skip}}
\newcommand{\passume}{\codekey{assume}}
\newcommand{\pif}{\codekey{if}}
\newcommand{\pthen}{\codekey{then}}
\newcommand{\pelse}{\codekey{else}}
\newcommand{\pwhile}{\codekey{while}}
\newcommand{\passert}{\codekey{assert}}
\newcommand{\passign}{:=}

\newcommand{\pderef}[2]{#1{\cdot}#2}
\newcommand{\pfree}[1]{\codekey{free}(#1)}
\newcommand{\palloc}[1]{\codekey{alloc}(#1)}

\newcommand{\exec}{\textsf{Exec}}
\newcommand{\pexec}{\textsf{PExec}}
\newcommand{\comp}{\textsf{Comp}}
\newcommand{\Terms}{\textsf{Terms}}

\newcommand{\rcong}{{\equiv}}

\newcommand{\init}[1]{\widehat{#1}}

\newcommand{\dblqt}[1]{\text{``}#1\text{''}}

\newcommand{\eqcl}[2]{[#1]_{#2}}

\newcommand{\undf}{\mathsf{undef}}

\newcommand{\pspc}{\mathsf{PSPACE}}

\newcommand{\locT}{\mathsf{Loc}}
\newcommand{\dataT}{\mathsf{Data}}
\newcommand{\sortT}{\mathsf{sort}}
\newcommand{\locTodataT}{{\locT \to \dataT}}
\newcommand{\fields}{\mathsf{Flds}}
\newcommand{\reachspec}{\varphi}
\newcommand{\startlocs}{\mathsf{Start}}
\newcommand{\reachptrs}{\mathsf{Pointers}}
\newcommand{\stoplocs}{\mathsf{Stop}}
\newcommand{\stoploc}{\mathsf{stop}}

\newcommand{\static}{\mathsf{static}}
\newcommand{\dynamic}{\mathsf{dynamic}}
\newcommand{\staticuniv}{U^\static_\locT}
\newcommand{\dynamicuniv}{U^\dynamic_\locT}

\newcommand{\val}{\mathsf{Val}^\Mm}
\newcommand{\fldval}{\fields\mathsf{Val}^\Mm}
\newcommand{\paramval}[1]{\mathsf{Val}^{#1}}
\newcommand{\paramfldval}[1]{\fields\mathsf{Val}^{#1}}

\newcommand{\allocs}{\mathsf{allocations}}

\newcommand{\alloc}{\mathsf{Alloc}}
\newcommand{\reach}{\mathsf{Reach}}

\newcommand{\memsafe}{\textsf{MS}}
\newcommand{\checkkoherence}{\textsf{checkSC}}
\newcommand{\infeasible}{q_\mathsf{infeasible}}
\newcommand{\unsafe}{q_\mathsf{unsafe}}

\newcommand{\nil}{\mathtt{NIL}}

\newcommand{\mahesh}{X}
\newcommand{\xtra}{\mahesh}

\newcommand{\coherent}{{\text{coherent}}}
\newcommand{\coherence}{{\text{coherence}}}
\newcommand{\Coherent}{{\text{Coherent}}}
\newcommand{\Coherence}{{\text{Coherence}}}

\newcommand{\alaw}{\text{alias-aware}}
\newcommand{\Alaw}{\text{Alias-Aware}}

\newcommand{\koherent}{{\text{streaming-coherent}}}
\newcommand{\Koherent}{{\text{Streaming-Coherent}}}
\newcommand{\koherence}{\text{streaming-coherence}}
\newcommand{\Koherence}{\text{Streaming-Coherence}}

\newcommand{\mapcomp}{\fields\textsf{Comp}}

\newcommand{\dynconst}{\mathsf{c}_\dynamic}
\newcommand{\dynptr}{\mathsf{f}_\dynamic}

\newcommand{\eqclosure}{\mathsf{Closure}^=}

\newcounter{exec}
\newcommand{\ctr}{\arabic{exec}}
\newcommand{\incctr}{\addtocounter{exec}{1}}

\newcommand{\heaps}{\textsf{Heaps}}

\newcommand{\benchfont}[1]{\textsf{#1}}

\begin{document}

\title[Deciding Memory
Safety for Single-Pass Heap-Manipulating Programs]{Deciding Memory Safety for Single-Pass Heap-Manipulating Programs}


\author{Umang Mathur}
\orcid{0000-0002-7610-0660}             
\email{umathur3@illinois.edu}          
\affiliation{
  \department{Department of Computer Science}             
  \institution{University of Illinois, Urbana-Champaign}           
  \country{USA}                   
}

\author{Adithya Murali}
\email{adithya5@illinois.edu} 
\affiliation{
  \department{Department of Computer Science}             
  \institution{University of Illinois, Urbana-Champaign}           
  \country{USA}                   
}

\author{Paul Krogmeier}
\email{paulmk2@illinois.edu}         
\affiliation{
  \department{Department of Computer Science}             
  \institution{University of Illinois, Urbana-Champaign}           
  \country{USA}                   
}

\author{P. Madhusudan}
\email{madhu@illinois.edu}         
\affiliation{
  \department{Department of Computer Science}             
  \institution{University of Illinois, Urbana-Champaign}           
  \country{USA}                   
}

\author{Mahesh Viswanathan}
\email{vmahesh@illinois.edu}         
\affiliation{
  \department{Department of Computer Science}             
  \institution{University of Illinois, Urbana-Champaign}           
  \country{USA}                   
}


\begin{abstract}
We investigate the decidability of automatic program verification for programs that manipulate heaps, and in particular, decision procedures for proving memory safety for them. We extend recent work that identified a decidable subclass of uninterpreted programs to a class of \alaw\ programs that can update maps. We apply this theory to develop verification algorithms for memory safety--- determining if a heap-manipulating program that allocates and frees memory locations and manipulates heap pointers does not dereference an unallocated memory location. We show that this problem is decidable when the initial allocated heap forms a forest data-structure and when programs are \emph{\koherent}, which intuitively restricts programs to make a single pass over a data-structure. Our experimental evaluation on a set of  library routines that manipulate forest data-structures shows that common single-pass algorithms on data-structures often fall in the decidable class, and that our decision procedure is efficient in verifying them.
\end{abstract}




\maketitle
\renewcommand{\shortauthors}{U. Mathur, A. Murali, P. Krogmeier, P. Madhusudan, M. Viswanathan}


\section{Introduction}
\seclabel{intro}

The problem of automatic (safety) verification is to ascertain whether a program satisfies its assertions on all inputs and on all executions.
The standard technique for proving programs correct involves writing
inductive invariants in terms of loop invariants and pre/post conditions,
and proving the resulting verification conditions valid~\cite{FloydMeaning1967,Hoare1969}.
While there has been tremendous progress in identifying decidable fragments
for checking validity of verification conditions (Nelson-Oppen combinations of decidable theories realized by efficient SMT solvers~\cite{calcofcomputation}),
decidable program verification without annotations
has been elusive. Apart from programs over finite domains, very few 
natural decidable classes are known.

In a recent paper~\cite{coherence2019}, a class of \emph{uninterpreted
programs} was identified and shown to have a decidable verification problem.
Uninterpreted programs work over arbitrary data domains; the domains give meaning to the constants, relations, and functions in the program, interpreting equality using its natural definition. A program is deemed correct only if it satisfies its assertions in all executions and for all data domains. The authors show that for
a class of programs that satisfies a \emph{\coherence} condition, verification
is decidable. The decision procedure relies on a streaming congruence closure
algorithm realized as  automata.
The results of~\cite{coherence2019} are, however, purely theoretical; there are no general application domains identified where uninterpreted program verification would be useful, nor is there any implementation.

The goal of this paper is to study completely automated verification for \emph{heap-manipulating programs}. When modeling the heap, we can treat pointers as \emph{unary uninterpreted functions}; this is a natural modeling choice that does not involve any abstraction (as pointer fields are really arbitrary, unary functions). Thus, it is reasonable to hope that uninterpreted program verification techniques could be useful in this setting.

However, there is a fundamental challenge that we need to overcome: heap-manipulating programs \emph{modify} heap pointers. With pointers modeled as functions, these programs modify functions/maps. It turns out that the theory of uninterpreted program verification developed in~\cite{coherence2019} is severely inadequate for tackling programs that modify maps. 

\emph{In this paper, we undertake a fundamental study of 
verification for 
programs that
have updatable maps/functions. We then apply this theory to show that checking memory safety of heap-manipulating programs (where one is given a recursively-described set of allocated locations and asked to check whether a program only dereferences locations that are within this set) is decidable
for a subclass of programs, and evaluate the algorithms with an implementation and experiments.}

\subsection{Alias Awareness and the Notion of Congruence for Programs with Updatable Maps}

We can think of uninterpreted programs as computating \emph{terms}
using the function and constants symbols appearing in the code---in the beginning, each variable stores a constant, and 
after executing an assignment statement of the form  $\dblqt{x \passign f(y)}$, the term stored in $x$ is $f(t_y)$,
where $t_y$ is the term stored 
in $y$ before executing this statement.
As the program executes, in addition to computing terms, it places equality and disequality constraints on the terms it computes (inherited from the conditionals in \pif-\pthen-\pelse~and \pwhile~constructs). 

However, when the program updates maps, this fundamental property, that the program computes terms, is destroyed. Consider an example where there are two locations pointed to by variables $x$ and $y$, but we do not know whether $x$ and $y$ point to the same location or not. If we update a pointer-field on $x$ and then read the same pointer-field from $y$ into a variable $k$, we cannot really associate any term with the variable $k$, as the term crucially depends on whether $x$ and $y$ were aliases to the same location or not.

The above is a fundamental problem that wreaks havoc, 
making program verification essentially 
impossible in the presence of updatable maps. 
The notion of \coherence\ in~\cite{coherence2019} fails, as it crucially relies on this notion of terms. And it cannot be repaired, as the semantic meaning of the \coherence~ condition fundamentally involves not computing the same term multiple times. 

The primary observation, which saves the framework, is to note that the entire problem is due to not knowing aliasing relationships. We call a program execution \emph{\alaw\,} if, at every point where it updates a function on the element pointed to by a program variable $x$, the precise aliasing relationship between $x$ and all other program variables in scope is \emph{determined}. More precisely, when the execution modifies a function/pointer-field of $x$, for every other program variable $z$, it must either be the case that $x$ is different from $z$ in any data-model/heap or it must be the case that $x$ is equal to $z$ in all data-models/heaps.

We show that alias-awareness is a panacea for our problems. For \alaw~ programs (programs whose executions are all \alaw), we show we can associate terms with variables after a computation that updates maps, and further show that the notion of \coherence~ extends naturally to programs that update maps. We then show that for coherent \alaw~ programs, the verification problem becomes decidable. These results constitute the first main contribution of the paper.

\vspace{-0.1in}
\subsection{Application to Verifying Memory Safety}

We then study the application of our framework to verifying memory safety. 
Our key observation is that for programs that manipulate 
\emph{forest data-structures} (data-structures consisting 
of disjoint tree-like structures), programs are naturally \alaw. 
Intuitively, when traversing forest data-structures, 
aliasing information is implicitly present. 
For instance, if $\cd{x}$ points to a location of 
a forest data-structure, we know that the location pointed to by 
$\cd{x}$, the one pointed to by the left child 
$\pderef{\cd{x}}{\cd{left}}$, and the one pointed to by the right child
$\pderef{\cd{x}}{\cd{right}}$ are all different.

In this paper, we define memory safety as follows. A heap-manipulating program starts with a set of allocated heap locations. During its execution, it dereferences pointers on heap locations, and allocates and frees locations. A program is memory safe if it never dereferences 
a location that is not in the allocated set. The above definition of memory safety
captures the usual categories of memory safety errors such as null-pointer dereferences, use after free, use of uninitialized memory, illegal freeing of memory, etc.~\cite{PLEnthusiast}.
However, in this paper, we do not consider allocation of \emph{contiguous blocks of arbitrary size} of memory (and hence do not handle arrays and buffer overflows of arrays in languages like C, etc.). Rather, we assume that allocation is done in terms of \emph{records} of fixed size (like \texttt{struct}s in C), and we disallow pointer arithmetic in our programs.

Our second main contribution of the paper is a technical result that gives an efficient decision procedure for verifying memory safety for a subclass of imperative heap-manipulating programs, whose \emph{initial\,} allocated heaps are restricted to forest data-structures.


Handling forest data-structures, i.e., disjoint lists and trees, is useful as they are ubiquitous. Note that we require only the initial
heaps to be forest data-structures;
the program can execute for an arbitrarily long time and create cycles/merges as it manipulates the structures. 
We model the primitive types and operations on them using uninterpreted functions and relations, similar in spirit to the way the work in ~\cite{coherence2019} handles all data. The key insight here is that this is a reasonable modeling choice, since programs typically do not rely on the semantics of the primitive data domains in order to assure memory safety (we also show this empirically in experiments).  
The salient aspect of our work is that we model the pointers in the heap \emph{precisely} using updatable maps, without resorting to \emph{any} abstraction (classical automatic analysis of heap programs
typically uses abstractions for heaps; for instance, shape analysis involves an abstraction of the heap locations to a finite abstract domain~\cite{tvla}).

We allow the user to specify the initial allocated set as the (unbounded) set of locations reachable from various locations (pointed to by certain program variables) using particular pointer fields, until a specific set of locations is reached.
The memory safety problem is then to check whether such a program, starting from an \emph{arbitrary} heap storing a forest data-structure, an arbitrary model for the primitive types, and with the specified allocated set, 
dereferences only those locations that 
are in the (potentially changing) allocated set, on all executions.

The above problem turns out to be undecidable (a direct consequence from~\cite{coherence2019}, as even programs that do not manipulate heaps and have simple equality assertions yield undecidability). 
The main result of the paper is that for a class of programs called \emph{\koherent\ programs}, memory safety is decidable. Intuitively, these correspond to programs that traverse the forest data-structures
in a \emph{single-pass}.

\paragraph{Technical Challenges.}

The primary challenge is dealing with \emph{updatable pointers} and \emph{updatable sets} (the latter are needed
to model the set of allocated locations, which changes during program execution). 
As we show in our first set of results, it is crucial for a program to be alias aware. The fact that our initial data-structures are forests implicitly causes \koherent~programs to be alias aware. When two variables point to locations obtained using different traversals, we know they cannot alias to each other. Also, for {\koherent} programs, we can keep track of whether traversals for any pair of variables are the same, and track their precise aliasing relationships.

\begin{sloppypar}
The culmination of the ideas above is our result that verification of memory safety for {\koherent} programs over forest data-structures is decidable, and is $\pspc$-complete. It is in fact decidable in time that is linear in the size of the program and exponential in the number of variables.
We also show that checking whether a given program is {\koherent} is decidable in $\pspc$.
Note that even checking reachability in programs with Boolean domains has this complexity, and hence our algorithms are quite efficient.
\end{sloppypar}

\paragraph{Evaluation.}
We implement a prototype of our automata-based decision procedure. This involves intersecting an automaton that checks whether a given \koherent~ execution is memory-safe with an automaton that represents the given program's executions. Instead of building these automata and checking emptiness of their intersection, we use an approach that constructs the automaton and the intersection on-the-fly, hence not paying the worst case costs upfront. We evaluate our procedure on a class of standard library functions that manipulate forest data-structures, including linked lists and trees, where various other aspects of the data-structures (such as keys, height, etc.) are modeled using an uninterpreted data domain.
These are typically \emph{single-pass} algorithms on such data-structures, that take pointers to forest data-structures as input
(and may create non-forest data-structures during computation). 

Though we have stringent requirements that programs must meet in order to be in the  decidable class, we show in our experiments that most  natural single-pass programs on forest data-structures meet our requirements. We also show that our tool is able to check if the program falls in the decidable class, and
both verify memory safety and find memory safety errors extremely efficiently. 

We emphasize that the novelty of our approach is in building \emph{decision procedures} for verifying memory safety without the aid of human-given loop invariants, and without abstracting the heap domain (the data
domain is, however, abstracted using uninterpreted functions). In contrast, there are several existing techniques that can prove memory safety \emph{when given manually written loop invariants} or prove memory safety by \emph{abstracting the heap} (which can lead to false positives). Our results hence carve out new ground in  memory safety verification and our experiments show that our approach holds promise for wider applicability and scalability.

\medskip

\noindent In summary, this paper makes the following contributions:
\begin{itemize}
    \item A notion of \alaw~ coherent programs (\secref{define-coherence}),
     and a result that shows that the assertion-checking
     problem for such programs is decidable and 
     $\pspc$-complete (\secref{alias-aware-coherent-verification}).
    \item A notion of \koherence~ and \emph{forest data-structures} (~\secref{forest}) with an efficient decision procedure for verifying memory safety for the class of \koherent\ programs that dynamically manipulate forest data-structures (\secref{automaton}).
    \item An efficient decision procedure determining if programs are \koherent~ (\secref{automaton}).
    \item An experimental evaluation (~\secref{experiments}) showing (a) common library routines that manipulate forest data-structures using single-pass traversals
    are often \koherent, and (b) that the decision procedures presented in this paper (for checking whether programs satisfy the
    \koherent\ requirement and for checking whether \koherent\ programs are memory safe)
    are very efficient, both for proving programs correct and finding errors in incorrect programs.
\end{itemize}

\section{Preliminaries}
\seclabel{prelim}

In this section, we define the syntax and semantics of
programs that manipulate heaps and other relevant concepts useful 
for presenting the main results of the paper.
These are fairly standard and familiar readers may skip the details.


\subsection{Syntax of Heap-Manipulating Programs}

The programs we consider are those that manipulate heaps. It is convenient to abstract heap
structures as consisting of two sorts of distinct elements --- a sort $\locT$ of memory 
locations in the heap, and a sort $\dataT$ of data values. \emph{Field} (or map) symbols
will model pointers from memory locations to memory locations or data values.
Constants and functions over the data domain will be used to construct other data values. In this paper, we will not assume any fixed interpretation 
for either data values or for functions on data values. In this sense, our
programs work over an \emph{uninterpreted} data domain. Predicates over data values will be 
modeled by functions capturing the characteristic function of the predicate. 

Let $\locT$ and $\dataT$ be the sorts of locations and data respectively.
Our vocabulary $\Sigma$ is a tuple of the form 
$(\Cc_\locT, \Ff_\locT, \Cc_\dataT, \Ff_\dataT, \Ff_\locTodataT)$,
where
\begin{itemize}
	\item $\Cc_\locT$ is a set of location constant symbols of sort $\locT$,
	\item $\Ff_\locT$ is a set of unary location function symbols with sort `$\locT, \locT$'~\footnote{We will
	use the notation $\sigma,\tau$ to indicate a function whose arguments are from sort $\sigma$ and which
	returns a value in sort $\tau$. Thus for example `$\locT,\locT$' is a function with one argument of sort
	$\locT$ and which returns an element of sort $\locT$. On the other hand, `$\dataT^r,\dataT$' denotes functions
	with $r$ arguments each of sort $\dataT$ and which returns an element of sort $\dataT$.}, that models pointers
	between heap locations,
	\item $\Cc_\dataT$ denotes the set of data constant symbols of sort $\dataT$,
	\item $\Ff_\dataT = \bigcup\limits_{i \geq 0} F_i$ where $F_r$ is a set of data function symbols of arity $r$ and sort `$\dataT^r, \dataT$', and
	\item $\Ff_\locTodataT$ is a set of unary location function symbols with sort `$\locT, \dataT$', modeling pointers to data
	values stored in heap locations.
\end{itemize}

\subsubsection{Program Syntax}
\seclabel{program-syntax}

Programs will use a finite set of variables to store information --- heap locations and data 
values --- during a computation.
Let us fix $V_\locT = \set{u_1, \ldots, u_l}$ as the set of
location variables
and $V_\dataT = \set{v_1, \ldots, v_m}$ as the set of data variables
and let $V = V_\locT \uplus V_\dataT$~\footnote{We use $A\uplus B$ to denote the 
disjoint union of sets $A$ and $B$.} be the set of all variables. 
In addition, our programs manipulate fields associated with
location variables.
We will model these fields as second order function variables
$\fields_\locT = \set{p_1, \ldots, p_r}$ (pointers from locations to locations) 
and $\fields_\dataT = \set{d_1, \ldots, d_s}$ (pointers from locations to data),
and let $\fields = \fields_\locT \cup \fields_\dataT$.
Taking $x, y \in V_\locT$, $p \in \fields_\locT$, $d \in \fields_\dataT$, $f$ in $\Ff_\dataT$, 
$a, b \in V_\dataT$, and $\vec{c}$ to be a tuple of variables in $V_\dataT$,
the syntax of
programs is given by the following grammar.
\begin{align*}
\stmt ::=& 
\,\,  \pskip \, 
\mid \, x \passign y \, 
\mid \, x \passign \pderef{y}{p} \, 
\mid \, \pderef{y}{p} \passign x \,
\mid \, a \passign \pderef{y}{d} \, 
\mid \, \pderef{y}{d} \passign a \,
\mid \, \palloc{x} \,
\mid \, \pfree{x}
\\
& \mid a \passign b \, 
\mid \, a \passign f(\vec{c}) \, 
\mid \, \passume \, (\cond) \,
\mid \, \stmt \, ;\, \stmt \\
&
\mid \, \pif \, (\cond) \, \pthen \, \stmt \, \pelse \, \stmt \,
\mid \, \pwhile \, (\cond) \, \stmt \\
\cond ::=& 
\, x = y \,
\mid \, a = b \,
\mid \, \cond \lor \cond \,
\mid \, \neg \cond
\end{align*}

Our programs have well-typed assignments to variables using values stored in other variables
($x \passign y$ and $a \passign b$) or using 
pointer dereferences from location variables, either to
the data sort ($a \passign \pderef{y}{d}$) or 
to the location sort ($x \passign \pderef{y}{p}$),
or using function computations in the data sort ($a \passign f(\vec{c})$).
Further, programs can update fields 
($\pderef{y}{d} \passign a$ or $\pderef{y}{p} \passign x$),
and can dynamically allocate ($\palloc{x}$) or deallocate ($\pfree{x}$) memory.
In addition, they allow the usual constructs of imperative 
programming --- empty statements ($\pskip$), conditionals ($\pif-\pthen-\pelse$)
and loops ($\pwhile$).
Conditionals in programs can be Boolean combinations of (well-typed)
equality atoms over location or data variables.

\subsection{Program Executions}
\seclabel{program-execution}

Executions of programs over $\Sigma$ and variables $V$ (given by the $\stmt$ grammar) 
are finite sequences over the alphabet $\Pi$
given below. 
\begin{align*}
\Pi = \setpred{\dblqt{ x \passign y }, \dblqt{x \passign \pderef{y}{p}}, \dblqt{\pderef{y}{p} \passign x}, \dblqt{a \passign \pderef{y}{d}}, \dblqt{\pderef{y}{d} \passign a}, \dblqt{\palloc{x}}, \dblqt{\pfree{x}}, \dblqt{a \passign b}, \\
\dblqt{a \passign f(\vec{c})}, \dblqt{\passume (x = y) }, \dblqt{\passume (x \neq y)},
\dblqt{\passume (a = b)}, \dblqt{\passume (a \neq b)}\\
}{x, y \in V_\locT, p \in \fields_\locT, d \in \fields_\dataT, f \in \Ff_\dataT
a, b \in V_\dataT, \text{ and }\vec{c} \text{ a tuple over }V_\dataT}.
\end{align*}

The set of executions of a program $P$, 
denoted $\exec(P)$ is given by a regular expression, inductively 
defined below. We assume that conditionals are in negation normal form
where $\dblqt{\passume(\neg(r = s))}$ 
translates to $\dblqt{\passume(r \neq s)}$
and $\dblqt{\passume(\neg(r \neq s))}$ 
translates to $\dblqt{\passume(r = s)}$.
\begin{align*}
\begin{array}{rcll}
\exec(\pskip) &=& \epsilon &  \\
\exec(a) &=& a & \text{if } a \in \Pi \\
\exec(\passume(c_1 \lor c_2)) &=& \exec(\passume(c_1)) + \exec(\passume(c_2)) & c_1, c_2 \in \cond \\
\exec(\passume(c_1 \land c_2)) &=& \exec(\passume(c_1))\cdot \exec(\passume(c_2)) & c_1, c_2 \in \cond \\
\exec(\pif \, (c) \, \pthen \, s_1 \, \pelse \, s_2) &=& 
\begin{aligned}
&\exec(\passume(c)) \cdot \exec(s_1) \\
&+ \, \exec(\passume(\neg c)) \cdot \exec(s_2)
\end{aligned} & 
\begin{aligned} 
&c \in \cond, \\ &s_1, s_2 \in \stmt
\end{aligned}\\
\exec(\pwhile\, (c)\, s) &=&
\begin{aligned}
&(\exec(\passume(c))\cdot\exec(s))^*\\
&\cdot\, \exec(\passume(\neg c))
\end{aligned} & 
\begin{aligned}
&c \in \cond, \\
&s \in \stmt
\end{aligned}\\
\exec(s_1;s_2) &=& \exec(s_1) \cdot \exec(s_2) & s_1, s_2 \in \stmt \\
\end{array}
\end{align*}
The set of partial executions of a program $P$, denoted $\pexec(P)$, is the set of
prefixes of its executions.

\subsection{Semantics of Executions}
\seclabel{prog-semantics}

The semantics of a heap-manipulating program is given
in terms of the behavior of its executions on \emph{heap structures}.

\subsubsection{Heap Structures}
\seclabel{heap-structure}
A $\Sigma$-heap structure is a tuple 
$\Mm = (U_\locT, U_\dataT, \Ii)$, where 
$U_\locT$ is a universe of locations,
$U_\dataT$ is a universe of data ($U_\locT \cap U_\dataT = \emptyset$)
and $\Ii$ is some interpretation of the various symbols
in $\Sigma$.
In order to faithfully model dynamic memory allocation,
we assume that the set of locations is the disjoint union of a statically allocated set
of locations and a \emph{countably infinite} set of locations that can be allocated dynamically. 
That is, we have $U_\locT = \staticuniv \uplus \dynamicuniv$,
where $\dynamicuniv = \set{e_0, e_1, \ldots}$ is an ordered
set of distinguished locations indexed by the set of natural numbers $\nats$.
The interpretation $\Ii$
maps every constant $c \in \Cc_\locT$ to an element from $\staticuniv$,
every constant in $\Cc_\dataT$ to an element from $U_\dataT$,
every function symbol 
$f \in \Ff_\locT$ to an element of $[U_\locT \to U_\locT]$,
symbol $f \in \Ff_\dataT$ of arity $r$ to an element of 
$[(U_\dataT)^r \to U_\dataT]$,
and, $f \in \Ff_\locTodataT$ to an element of $[U_\locT \to U_\dataT]$.
Further, we assume that the elements in $\dynamicuniv$
cannot be accessed from $\staticuniv$, i.e.,
for every $f \in \Ff_\locT$, we have
$\forall e \in \staticuniv, \Ii(f)(e) \in \staticuniv$. This ensures that the set of locations reachable by any execution cannot access locations in the dynamic universe that have not been allocated by the execution. Moreover, for every $f \in \Ff_\locT$ and every $e \in \dynamicuniv$, we have $f(e) = e$. Finally, we assume in $\Sigma$ the presence of a distinguished
constant symbol  $\dynconst$ and a distinguished
function symbol $\dynptr$
not used in the syntax $\stmt$ of programs.
We require that
for every heap structure $\Mm = (\staticuniv \uplus \dynamicuniv, U_\dataT, \Ii)$
the interpretation function $\Ii$ of $\Mm$
assigns interpretations to these symbols such that
$\Ii(\dynconst) = e_0$ and $\Ii(\dynptr) \in [\dynamicuniv \to \dynamicuniv]$
with $\Ii(\dynptr)(e_i) = e_{i+1}$ for every $i \in \nats$.

\subsubsection{Valuation of Variables and Pointers in an Execution.}

We assume that corresponding to each program variable $x\in V$,
there is a distinguished constant $\init{x} \in \Cc_\locT \uplus \Cc_\dataT$
of the appropriate sort denoting the initial value of $x$.
Likewise, each field $p \in \fields_\locT$ 
(resp. $d \in \fields_\dataT$) is also associated with
a unary function $\init{p} \in \Ff_\locT$ (resp. $\init{d} \in \Ff_\locTodataT$).

Given an execution $\sigma \in \Pi^*$ of a program $P \in \stmt$
and a $\Sigma$-heap structure $\Mm = (U_\locT, U_\dataT, \Ii)$, the valuation of
the program variables and field pointers at the end of $\sigma$ are defined 
in terms of valuation functions 
$\val_\locT : \Pi^* \times V_\locT \to U_\locT$,
$\val_\dataT : \Pi^* \times V_\dataT \to U_\dataT$,
$\fldval_\locT : \Pi^* \times \fields_\locT \to [U_\locT \to U_\locT]$
and $\fldval_\dataT : \Pi^* \times \fields_\dataT \to [U_\locT \to U_\dataT]$,
which are presented next. 
In the following, $\allocs(\sigma)$ denotes
the number of occurrences of statements of the form $\dblqt{\palloc{\cdot}}$ in $\sigma$.

\begin{align*}
\begin{array}{rcl}
\val_\locT(\varepsilon, u) &=& \Ii(\init{u}) \\ 
\val_\locT(\sigma \cdot s, u) &=& 
\begin{cases} 
\val_\locT(\sigma, y) & \text{if } s = \dblqt{x \passign y} \text{ and } u = x \\
e_{i-1} & \text{if } s = \dblqt{\palloc{x}}, i = \allocs(\sigma) \\
&\text{ and } u = x \\
\fldval_\locT(\sigma, p)(\val_\locT(\sigma, y)) & \text{if } s = \dblqt{x \passign \pderef{y}{p}} \text{ and } u = x \\  
\val_\locT(\sigma, u) & \text{otherwise}
\end{cases}
\end{array}
\end{align*}

\begin{align*}
\begin{array}{rcl}
\val_\dataT(\varepsilon, v) &=& \Ii(\init{v}) \\ 
\val_\dataT(\sigma \cdot s, v) &=& 
\begin{cases} 
\val_\dataT(\sigma, b) & \text{if } s = \dblqt{a \passign b} \text{ and } v = a \\
\Ii(f)(\val_\dataT(\sigma, c_1), \ldots, \val_\dataT(\sigma, c_r)) & \text{if } s = \dblqt{a \passign f(c_1, \ldots, c_r)} \\
& \text{ and } v = a \\
\fldval_\dataT(\sigma, d)(\val_\locT(\sigma, y)) & \text{if } s = \dblqt{a \passign \pderef{y}{d}} \text{ and } v = a \\  
\val_\dataT(\sigma, v) & \text{otherwise}
\end{cases}
\end{array}
\end{align*}

\begin{align*}
\begin{array}{rcl}
\fldval_\locT(\varepsilon, p) &=& \Ii(\init{p}) \\
\fldval_\locT(\sigma\cdot s, p) &=& 
\begin{cases}
\fldval_\locT(\sigma, q)[\val_\locT(\sigma, y) \mapsto \val_\locT(\sigma, x)] & \text{ if } s = \dblqt{\pderef{y}{q} \passign x} \\
& \text{ and } p = q\\
\fldval_\locT(\sigma, p) & \text{otherwise}
\end{cases}
\end{array}
\end{align*}

\begin{align*}
\begin{array}{rcl}
\fldval_\dataT(\varepsilon, d) &=& \Ii(\init{d}) \\
\fldval_\dataT(\sigma\cdot s, d) &=& 
\begin{cases}
\fldval_\dataT(\sigma, h)[\val_\locT(\sigma, y) \mapsto \val_\dataT(\sigma, a)] & \text{ if } s = \dblqt{\pderef{y}{h} \passign a}\\
& \text{ and } d = h\\
\fldval_\dataT(\sigma, d) & \text{otherwise}
\end{cases}
\end{array}
\end{align*}

where by $f[a \mapsto b]$ we mean the function $g$ defined as $g(a) = b$ and $g(x) = f(x)$ otherwise.


\paragraph{Feasibility.} An execution $\sigma$ is said to be \textbf{feasible} on $\Mm$
if for every prefix of $\sigma$ of the form 
$\rho' \cdot \dblqt{\passume(x = y)}$,
we have 
$\val_\sortT(\rho', x) = \val_\sortT(\rho', y)$,
and
for every prefix of $\sigma$ of the form 
$\rho' \cdot \dblqt{\passume(x \neq y)}$,
we have 
$\val_\sortT(\rho', x) \neq \val_\sortT(\rho', y)$,
where $\sortT \in \set{\locT, \dataT}$ is the sort of both $x$ and $y$.


\section{Defining \Coherence~ For Heap-Manipulating Programs}
\seclabel{define-coherence}

In this section we discuss some of the challenges involved in
coming up with a reasonable extension for the notion of
\coherence~ as defined in~\cite{coherence2019}.
A key problem in defining such an extension, and indeed in generally handling programs with updatable maps, lies in
keeping accurate track of \emph{aliasing} between 
variables of the location sort. We will first discuss some examples that highlight these challenges and proffer a first solution to the problem. Then we will discuss our solution more formally by introducing relevant notation and define our notion of \emph{\alaw}~
executions and programs that captures the essence of the aliasing problem. Finally, we will discuss the notion of \coherence~ adapted
to the case of heap-manipulating programs.


\subsection{The Importance of Being Alias Aware for Programs Updating Maps}
\seclabel{examples-challenges}

Functions and updatable maps cannot be handled uniformly; in particular, the work of~\cite{coherence2019} does not immediately lend itself to handling updatable maps. Let us illustrate this using the following example.

\begin{example} 
\incctr
Consider a straight-line program that generates execution $\pi_1$, where
\begin{align*}
\pi_\ctr \delequal
\cd{z}_1 \passign \pderef{\cd{x}}{\cd{next}}
\cdot\, 
\passume(\cd{z}_1 \neq \cd{z}_2)
\cdot\,
\pderef{\cd{y}}{\cd{next}} \passign \cd{z}_2
\cdot\,
\cd{z}_3 \passign \pderef{\cd{x}}{\cd{next}}
\end{align*}
In the execution above we do not, and in fact cannot, know the value stored in the variable $\cd{z}_3$ unless we know whether $\cd{x} = \cd{y}$ or $\cd{x} \neq \cd{y}$. In the former case, we will have that $\cd{z}_3 = \cd{z}_2$, since $\cd{y}$ \emph{aliases} $\cd{x}$ and, therefore, the update of the $\cd{next}$ pointer on $\cd{y}$ will have over-written the value of the pointer on $\cd{x}$. Similarly, in the latter case we will have that $\cd{z}_3 = \cd{z}_1$, since the pointer update on $\cd{y}$ will have no effect on $\cd{x}$.

This is a major difference, since in any heap structure on which $\pi_\ctr$ is feasible,
it must be that $\cd{z}_1 \neq \cd{z}_2$. Therefore, 
whether $\cd{x}$ and $\cd{y}$ are aliases of each other
can have a drastic effect on the semantics and feasibility 
of the execution. 
\end{example}

The above example illustrates that aliasing plays a crucial role in the semantics and feasibility of executions. There is, however, a second (related) issue as well.

\begin{example} 
\incctr
\incctr
Consider the executions $\pi_2$ and $\pi_3$ where
\begin{align*}
&\pi_2 \delequal \pi_1 \cdot \passume(\cd{z}_2 = \cd{z}_3)
\,\,\,\,\,\,\,
\,\,\,\,\,\,\, \pi_3 \delequal \pi_1 \cdot \passume(\cd{z}_2 \neq \cd{z}_3)
\end{align*}
\end{example}

As discussed earlier, execution $\pi_2$ is only feasible in models where $\cd{x} = \cd{y}$ and $\pi_3$ only in models where $\cd{x} \neq \cd{y}$ where at the end of $\pi_1$ both kinds of models were feasible. Therefore one can have \emph{implied} equalities and disequalities between variables that are only known much later in the execution. In general, this is hard to keep track of in a streaming setting (which we wish to do in order to use the underlying ideas in~\cite{coherence2019}) and can require an unbounded amount of memory, as can be seen in the following example.

\begin{example}
\incctr
Consider the following execution
\begin{align*}
\pi'_\ctr \delequal
&\passume(\cd{x} \neq \nil)
\cdot
\passume(\cd{y} \neq \nil)
\cdot
\cd{z}_1 \passign \pderef{\cd{x}}{\cd{next}}
\cdot 
\passume(\cd{z}_1 \neq \cd{z}_2)
\cdot
\pderef{\cd{y}}{\cd{next}} \passign \cd{z}_2 \\
&\cdot
\cd{z}_3 \passign \pderef{\cd{x}}{\cd{next}}
\cdot
\cd{k}_1 \passign \pderef{\cd{x}}{\cd{key}}
\cdot 
\cd{k}_2 \passign \pderef{\cd{y}}{\cd{key}}
\cdot 
\underbrace{
\cd{k}_1 \passign \cd{f}(\cd{k}_1) \cdot \cd{k}_2 \passign \cd{f}(\cd{k}_2) \cdots \cd{k}_1 \passign \cd{f}(\cd{k}_1) \cdot \cd{k}_2 \passign \cd{f}(\cd{k}_2)
}_{n}\\
& \cdot
\passume(\cd{z}_2 = \cd{z}_3)
\cdot
\passume(\cd{k}_1 = \cd{k}_2)
\end{align*}
\end{example}

The above execution is feasible, but would require an unbounded amount of memory to reason as such in a streaming fashion. Next we will see a solution to this aliasing problem that will allow us to keep track of relationships between variables and the correct semantics of an execution.

\subsection{Alias-Aware Executions and Programs}
\seclabel{alias-aware-def}

In~\secref{examples-challenges}, we saw that unlike the work of~\cite{coherence2019} we cannot define what term (or value) is computed by a variable. 
We also observed that the main issue with programs that have updatable maps is \emph{aliasing} --- i.e., when a pointer or data field is updated on a variable, the update may also be true for a different variable that \emph{aliases} the original variable. 
In this section we identify a class of executions, called
\alaw~executions, which implicitly resolves any aliasing when
updating the pointer fields on heap locations.

Below, we formally define \alaw~executions.
We denote by $\heaps(\sigma)$ the set of heap structures
on which the execution $\sigma$ is feasible.
\begin{definition}
\deflabel{alias-aware}
Let $\rho \in \Pi^*$ be an execution. $\rho$ is said to be \alaw~if for every prefix of $\rho$
of the form $\sigma \cdot \dblqt{\pderef{x}{h} \passign u}$
(the sorts of $u$ and $h$ are compatible),
and for every location variable $y \in V_\locT$, 
one of the following hold
\begin{itemize}
	\item For every heap structure $\Mm = (U_\locT, U_\dataT, \Ii) \in \heaps(\sigma)$, $\val_\locT(\sigma, x) = \val_\locT(\sigma, y)$.
	\item For every heap structure $\Mm = (U_\locT, U_\dataT, \Ii) \in \heaps(\sigma)$, $\val_\locT(\sigma, x) \neq \val_\locT(\sigma, y)$.
\end{itemize}
\end{definition}

Intuitively, the above definition says that for an \alaw~ execution, at a map update we know all the relevant aliasing information (even in the uninterpreted sense), because the variables that alias to the variable on which the update is being performed are the same in \emph{every} feasible model.

\begin{definition}[\Alaw~ programs]
A program $s \in \stmt$ is said to be \alaw~ if every execution $\rho \in \exec(s)$ is \alaw.
\end{definition}

\subsection{Term Computed by a Variable during an Execution}

Before we define \coherence~in our setting,
we require the notion of \emph{terms} that an execution
computes. 
The notion of terms lets us reason about infinitely many
heap structures in a  symbolic fashion and was
crucially exploited in~\cite{coherence2019} to define
\coherence~and in arguing the correctness of the
decision procedure designed for \coherent~programs.

We would like to define the term associated with a program variable
in an execution. There are two major challenges in our way. The first challenge is that, in order to accurately determine
the \emph{value} pointed to by a variable during an execution on a concrete heap structure,
one needs to accurately keep track of 
the interpretations of pointer fields (that could have, in turn, been updated earlier in the execution). A similar problem needs to be addressed in order to successfully define the notion of terms. The second challenge is that unlike in concrete heap structures where two distinct elements
in the heap are known to be unequal, terms do not behave the same way.
For the case of terms, we would like to explore
the possibility that two terms that might be syntactically
different might still be semantically \emph{equivalent}.
While this subtlety can be dealt with easily when none of the
functions are updatable, as in~\cite{coherence2019}, greater care is required in our setting.
Let us first formalize some notation and then elaborate
these subtleties.

For a set $\Cc$ of constant symbols and a set $\Ff$
of function symbols from $\Sigma$,
we use $\Terms(\Cc, \Ff)$ (or just $\Terms$ when the signature is clear) to denote the
set of (well-sorted) ground terms constructed using the constants 
in $\Cc$ and function symbols in $\Ff$.

For a binary relation $R \subseteq \Terms\times \Terms$,
the \emph{congruence closure} of $R$, denoted $\cong_R$
is the smallest equivalence relation such that 
(i) $R \subseteq \cong_R$, and 
(ii) for every function $f$ of sort $w_1w_2\cdots w_r,w$ ($w, w_i \in \set{\locT, \dataT}$)
	and terms $t_1, t'_1, t_2, t'_2, \ldots, t_r, t'_r$
	of sorts $w_1, w_1, w_2, w_2, \ldots, w_r, w_r$, 
	we have
	$\Big( \bigwedge\limits_{i=1}^r (t_i, t'_i) \in \cong_R \Big)
	{\implies}(f(t_1, \ldots, t_r), f(t'_1, \ldots, t'_r)) \in \cong_R$. 
We say $t_1 \cong_{R} t_2$ if $(t_1, t_2) \in \cong_R$, where $t_1,t_2$ are terms.

We now define the terms associated with variables
$\comp : \Pi^* \times V \to \Terms$,
the interpretations of updatable maps on terms corresponding to variables
$\mapcomp : \Pi^* \times \fields \times V \to \Terms$, and the set of equalities accumulated by the execution
$\alpha : \Pi^* \to \Pp(\Terms \times \Terms)$.
Recall that, in the following, $\dynconst \in \Cc_\locT$ and 
$\dynptr \in \Ff_\locT$ are special symbols in our vocabulary.

\begin{align*}
\begin{array}{rcl}
\comp(\varepsilon, u) &=& \init{u} \\
\comp(\sigma \cdot s, u) &=& 
\begin{cases} 
\comp(\sigma, y) & \text{if } s = \dblqt{x \passign y} \text{ and } u = x \\
\dynptr^i(\dynconst) & 
\begin{aligned}
&\text{if } s = \dblqt{\palloc{x}}, \\
&i = \allocs(\sigma), \text{ and } u = x 
\end{aligned}\\
\mapcomp(\sigma, h, y) & 
\text{if } s = \dblqt{x \passign \pderef{y}{h}} \text{ and } u = x \\  
f(\comp(\sigma, z_1), \ldots, \comp(\sigma, z_r)) & \text{if } s = \dblqt{x \passign f(z_1, \ldots, z_r)} \text{ and } u = x \\  
\comp(\sigma, u) & \text{otherwise}
\end{cases}
\end{array}
\end{align*}

\begin{align*}
\begin{array}{rcl}
\mapcomp(\varepsilon, h, z) &=& \init{h}(\init{z}) \\
\mapcomp(\sigma \cdot s, h, z) &=& 
\begin{cases} 
\comp(\sigma, x) & 
\begin{aligned}
\begin{array}{l}
\text{ if } s = \dblqt{\pderef{y}{h} \passign x} \\ 
\text{ and } \comp(\sigma, z) \cong_{\alpha(\sigma)} \comp(\sigma, y) 
\end{array}
\end{aligned}
\\
\mapcomp(\sigma, h, z) & \text{otherwise}
\end{cases}
\end{array}
\end{align*}

\begin{align*}
\begin{array}{rcl}
\alpha(\varepsilon) &=& \emptyset \\
\alpha(\sigma \cdot s) &=& 
\begin{cases} 
\alpha(\sigma) \cup \set{(\comp(\sigma, x), \comp(\sigma, y))} & \text{if } s = \dblqt{\passume(x = y)} \\
\alpha(\sigma) & \text{otherwise}
\end{cases}
\end{array}
\end{align*}

The set of terms computed by an execution $\sigma$
is the set $\Terms(\sigma) = \setpred{\comp(\rho, x)}{x \in V, \rho \text{ is a prefix of } \sigma}$.

Let us discuss some aspects of the above definition here.
The effect of some of the commands in the executions is
obvious and similar to~\cite{coherence2019}.
At the beginning of an execution, each term $x \in V$
is associated with a unique constant term $\init{x}$,
each pointer $h$ assigns every term $t$ 
(of the location sort) to $\init{h}(t)$,
and the set of equalities accumulated is $\emptyset$.
On an assignment $\dblqt{x \passign y}$,
the term stored in $x$ is updated to be that stored in $y$.
On an assignment involving a (non-updatable) function
($x \passign f(z_1 ,\ldots, z_r)$), the execution computes the term
$f(t_1, \ldots, t_r)$, and stores it in $x$, where $t_i$ is the term
associated with $z_i$ before the assignment.

Let us now look at the other commands.
Recall that we model our set of dynamically allocated 
locations as a disjoint set from the set of statically allocated
locations.
The term associated with $x$ after an $\dblqt{\palloc{x}}$
follows this premise. It is built using $n$ applications
of a distinct function $\dynptr \in \Ff_\locT$ 
on a distinct constant $\dynconst \in \Cc_\locT$,
where $n$ is the number of allocations performed before
this point in the execution.
Let us now turn to the most subtle aspect: \emph{pointer updates}.
When the execution encounters a command of the form
$\dblqt{\pderef{y}{h} \passign x}$, we need to not only
update the term pointed to by $\pderef{t_y}{h}$ 
(where $t_y$ is the term associated with $y$),
we also need to update the term pointed to by
$\pderef{t}{h}$, when $t$ can be inferred to be equivalent
to $t_y$ using the equalities $\alpha(\cdot)$ accumulated so far.

\begin{example}
\exlabel{terms_subtlety}
\incctr
Consider the following execution. 
\begin{align*}
\pi_\ctr \delequal
\passume(\cd{x} = \cd{y})
\cdot
\pderef{\cd{x}}{\cd{p}} \passign \cd{z}_1
\cdot
\cd{z}_2 \passign \pderef{\cd{y}}{\cd{p}}
\cdot
\passume(\cd{z}_1 \neq \cd{z}_2)
\end{align*}
Here, the terms associated with the pointers $\pderef{x}{p}$
and $\pderef{y}{p}$ are $\init{p}(\init{x})$ and
$\init{p}(\init{y})$ respectively.
After the $\dblqt{\passume(\cd{x} = \cd{y})}$,
the locations pointed to by $\cd{x}$ and $\cd{y}$ are the same
in every heap structure on which this assumption is valid,
and thus the terms corresponding to $\cd{x}$ and $\cd{y}$
must be deemed \emph{equivalent}.
This means that the pointer update $\pderef{\cd{x}}{\cd{p}} \passign \cd{z}_1$ 
should in fact be reflected in the term associated with the
pointer $\pderef{\cd{y}}{\cd{p}}$.
The above definition of $\comp(\cdot)$ and $\mapcomp(\cdot)$
will, in fact, ensure that
$\comp(\pi_\ctr, \cd{z}_1) = \comp(\pi_\ctr, \cd{z}_2)$.
Not doing so will, in turn, not reflect the contradiction
due to the last statement of this execution,
which makes it infeasible in every heap structure. 
\end{example}

It turns out that the above definition of terms associated with
variables is only a \emph{best effort}, in that,
it may not accurately summarize all heap 
structures in which the execution
is feasible.
The following example illustrates why this is the case.

\begin{example}
\exlabel{terms_incorrect}
\incctr
Let us consider the following, which
is a permutation of the execution in~\exref{terms_subtlety}.
\begin{align*}
\pi_\ctr \delequal
\pderef{\cd{x}}{\cd{p}} \passign \cd{z}_1
\cdot
\cd{z}_2 \passign \pderef{\cd{y}}{\cd{p}}
\cdot
\passume(\cd{x} = \cd{y})
\cdot
\passume(\cd{z}_1 \neq \cd{z}_2)
\end{align*}
Similar to the execution in~\exref{terms_subtlety},
the execution $\pi_\ctr$ is infeasible in every heap structure.  
However, in this case,
the terms associated with $\cd{z}_1$ and $\cd{z}_2$
will be 
$\comp(\pi_\ctr, \cd{z}_1) = \init{\cd{z}_1}$
and  $\comp(\pi_\ctr, \cd{z}_2) = \init{\cd{p}}(\init{\cd{y}})$,
which are not equal,
even when considering the equivalence
induced due to $\alpha(\pi_\ctr) = \set{(\init{\cd{x}}, \init{\cd{y}})}$.
As a result, it is hard to conclude that
$\pi_\ctr$ is an infeasible execution by analyzing the
terms associated with variables alone.
\end{example}

However, for the case of an \alaw~execution (\defref{alias-aware}), 
the definition of terms above indeed accurately relates the terms associated with each variable to their values in every heap structure that makes the execution
feasible. This brings us to the first main result of the paper that motivates the definition of \alaw~executions. It simply states that the interpretation of the terms defined above on a given model coincides with the actual values computed on the model.

\begin{theorem}
\thmlabel{terms-sanity}
Let $\sigma$ be an \alaw~ execution and let $\Mm = (U_\locT, U_\dataT, \Ii) \in \heaps(\sigma)$ be a heap structure on which $\sigma$ is feasible.
Then, $\Ii(\comp(\sigma, x)) = \val_\locT(\sigma, x)$ for every variable $x \in V_\locT$ and
$\Ii(\comp(\sigma, a)) = \val_\dataT(\sigma, a)$ for every $a \in V_\dataT$. 

Moreover, $\Ii(\mapcomp(\sigma,h,x)) = \fldval_\locT(\sigma,h)(\val_\locT(\sigma,x))$ and $\Ii(\mapcomp(\sigma,d,x)) = \fldval_\dataT(\sigma,d)(\val_\locT(\sigma,x))$ for every $h \in \fields_\locT, d \in \fields_\dataT$ and $x \in V_\locT$.
\end{theorem}

\begin{proof}
The proof proceeds by induction on the length of the execution. The base case and most inductive cases are trivial by the definitions of $\comp,\val_\locT,\val_\dataT$, etc. The only nontrivial case is that of map update. Consider an execution of the form $\sigma \cdot \dblqt{\pderef{y}{h} \passign x}$. It is easy to see that definitions of $\fldval_\dataT$ and $\fldval_\locT$ update the map on the location pointed to by $y$. However, the definition of $\comp$ updates the appropriate term for the map $h$ only on variables whose terms are equivalent modulo $\alpha(\sigma)$. We simply establish that for \alaw~ executions the equivalence modulo $\alpha(\sigma)$ is sufficient to determine which variables alias $y$. This is owing to the fact that the initial model of terms modulo the congruence closure of $\alpha(\sigma)$ is feasible on $\sigma$. And in that model, the only variables that alias $y$ are those that compute terms equivalent to $\comp(\sigma,y)$ modulo $\alpha(\sigma)$. Since $\sigma$ is \alaw, we have that in any model the variables that alias to $y$ are the same as those that alias to $y$ in the initial model.

We establish that the initial model of an \alaw~ execution $\sigma$ is a feasible model of $\sigma$, by strong induction on the length of the execution, where we strengthen our inductive hypothesis. We present the stronger inductive hypothesis below, the proof of which is a simple mutual (structural) induction on the execution.

Given an \alaw~ execution $\sigma$ that is feasible (i.e., feasible in some model) and a sub-execution $\rho \leq \sigma$:
\begin{itemize}
    \item $\rho$ is feasible on $\Mm_{\alpha(\rho)}$, where $\Mm_{\alpha(\rho)}$ is the initial term modulo the congruence closure of $\alpha(\rho)$.
    \item $\paramval{\Mm_{\alpha(\rho)}}(\rho, x) = \eqcl{\comp(\rho,x)}{\alpha(\rho)}$, where $\eqcl{}{\alpha(\rho)}$ represents the equivalence class of terms modulo the congruence defined by $\alpha(\rho)$. 
    
    Similarly, $\paramfldval{\Mm_{\alpha(\rho)}}_\locT(\rho, h)(\paramval{\Mm_{\alpha(\rho)}}_\locT(\rho, x)) = \eqcl{\mapcomp(\rho,h,x)}{\alpha(\rho)}$
    
    and $\paramfldval{\Mm_{\alpha(\rho)}}_\dataT(\rho, d)(\paramval{\Mm_{\alpha(\rho)}}_\locT(\rho, x)) = x\eqcl{\mapcomp(\rho,d,x)}{\alpha(\rho)}$
    \item $\Mm_{\alpha(\sigma)}$ is feasible on $\rho$.
    \item $\paramval{\Mm_{\alpha(\sigma)}}(\rho, x) = \eqcl{\comp(\rho,x)}{\alpha(\sigma)}$
    
    Similarly, $\paramfldval{\Mm_{\alpha(\sigma)}}_\locT(\rho, h)(\paramval{\Mm_{\alpha(\sigma)}}_\locT(\rho, x)) = \eqcl{\mapcomp(\rho,h,x)}{\alpha(\sigma)}$
    
    and $\paramfldval{\Mm_{\alpha(\sigma)}}_\dataT(\rho, d)(\paramval{\Mm_{\alpha(\sigma)}}_\locT(\rho, x)) = \eqcl{\mapcomp(\rho,d,x)}{\alpha(\sigma)}$
\end{itemize}
\end{proof}




\subsection{\Coherent~Programs}
\seclabel{coherence-prelim}

Having defined the concept of terms associated with
variables in an execution, we can now extend the
notion of \emph{coherence}, adapting it from~\cite{coherence2019}
to accommodate updatable pointers as in the heap-manipulating
programs that we study.

In the following, we say that $s$ is a superterm of $t$
modulo some congruence relation $\cong$
if there are terms $s'$ and $t'$ such that
$t \cong t'$, $s'$ is a superterm of $t'$
and $s' \cong s$.
\begin{definition}[\Coherence]
\deflabel{coherence}
A complete or partial execution $\sigma$ is said to be a \coherent~ execution if it satisfies the  following two conditions.
\begin{description}
    \item[\textbf{Memoizing.}]
    Let $\rho$ be a prefix of
    $\sigma$ such that it is of the form 
    $\rho'\cdot \dblqt{x \passign \pderef{y}{h}}$ or 
    $\rho'\cdot \dblqt{x \passign f(\vec{y})}$ and let 
    $t_x = \comp(\rho,x)$. 
    If there is a term $t' \in \Terms(\rho')$ such that $t' \cong_{\alpha(\rho')} t_x$, then it must be 
    the case that there is some variable $z \in V$ such that 
    $\comp(\rho',z) \cong_{\alpha(\rho')} t_x$.
    
    \item[\textbf{Early Assumes.}] 
    Let $\rho$ be a prefix of the form 
    $\rho'\cdot \dblqt{\passume(x = y)}$ and 
    let $t_x = \comp(\rho',x), t_y = \comp(\rho',y)$. 
    If there is a term $t \in \Terms(\rho')$ 
    such that $t$ is a superterm of 
    $t_x$ or $t_y$ modulo $\cong_{\alpha(\rho')}$, 
    then there must be some variable $w \in V$ such that 
    $\comp(\rho',w) \cong_{\alpha(\rho')} t$.
\end{description}
Note that these conditions are applicable to every sort-sensible combination of symbols.
\end{definition}

As mentioned earlier, this definition of \coherence~ is inspired from 
the notion of coherence defined previously in~\cite{coherence2019}.
The \emph{memoizing} restriction is the heart of the notion of \coherent~executions. 
Let us illustrate with a simple example.
\begin{example}
\exlabel{memoizing}
\incctr
Let $\pi_\ctr$ be the following execution:
\begin{align*}
\pi_\ctr \delequal 
 \cd{u} \passign \cd{f(w)}
\cdot \underbrace{\cd{u} \passign \cd{f(u)} \cdots \cd{u} \passign \cd{f(u)}}_{n}
\cdot \cd{v} \passign \cd{f(w)}
\cdot \underbrace{\cd{v} \passign \cd{f(v)} \cdots \cd{v} \passign \cd{f(v)}}_{n}
\cdot \passume (\cd{u} \neq \cd{v})
\end{align*}
It is easy to see that in the above execution $\pi_\ctr$,
the values of the variables $\cd{u}$ and \cd{v} will be equal
after executing all the $2n + 2$ assignments, in any heap structure.
As a result of the last assumption $\dblqt{\passume (\cd{u} \neq \cd{v})}$,
the execution
$\pi_\ctr$ will thus be infeasible in all heap structures.

However, in order to accurately determine the 
relationship $\dblqt{\cd{u} = \cd{v}}$ at the end of the execution, 
one needs to keep track of an unbounded amount of 
information (as $n$ can be increased unboundedly).
This is the crucial insight in the first condition (\emph{memoizing}),
which states that when a term has been computed and \emph{dropped}
(i.e., no variable points to the term),
then the execution should not recompute this term.
Indeed, the above execution $\pi_\ctr$ does not meet
this requirement --- the execution computed
the term $t = \cd{f}(\init{\cd{w}})$
in its first half and stored it in $\cd{u}$,
re-assigned this variable $\cd{u}$ immediately after
computing $t$, thereby losing all copies of $t$,
and then later recomputed this term again in its second half.
In fact, each of the terms $\cd{f}^i(\cd{w})$ ($1\leq i < n$)
have been recomputed without retaining their original copies.
Clearly, $\pi_\ctr$ is not a coherent execution.
\end{example}

A similar example highlights the importance of the second coherence restriction (\emph{early} assumes).
\begin{example}
\exlabel{early-assume}
\incctr
Let $\pi_\ctr$ be the following execution:
\begin{align*}
\pi_\ctr \delequal 
\cd{u} \passign \cd{u}_0
\cdot \overbrace{\cd{u} \passign \cd{f(u)} \cdots \cd{u} \passign \cd{f(u)}}^{n}
\cdot \cd{v} \passign \cd{v}_0
\cdot \overbrace{\cd{v} \passign \cd{f(v)} \cdots \cd{v} \passign \cd{f(v)}}^{n}
&\cdot \passume (\cd{u}_0 = \cd{v}_0)\\
&\qquad\qquad \cdot \passume (\cd{u} \neq \cd{v})
\end{align*}
Observe that this execution is similar to the execution in~\exref{memoizing}.
Also observe that, as before, this execution is infeasible in any heap structure.
However, any algorithm that would accurately determine 
this in a \emph{streaming} fashion would require access to unbounded memory.
\end{example}

\begin{definition}
A program is said to be \coherent~if all its executions are \coherent.
\end{definition}



\section{Assertion Checking for Coherent \Alaw~ Programs}
\seclabel{alias-aware-coherent-verification}

In this section we discuss our first main result --- decidability
of assertion checking for programs that are both \alaw~ as well
as coherent.

Let us first define the problem of assertion checking.
For this, we augment our programs with a special statement
`\passert(\pfalse)', and thus our new syntax is 
given by the following grammar ($\stmt$ is as defined in~\secref{program-syntax}).
\begin{align*}
\stmt_\text{assert} ::=& 
\,\,  \stmt \, 
\mid \, \passert(\pfalse) \, 
\mid \, \stmt_\text{assert} \, ;\, \stmt_\text{assert}
\end{align*}
Observe that more complex assertions 
(including boolean combinations of equality/disequality assertions) 
can be expressed by translating them to conditional statements.
For example, the assertion `$\passert(x = y)$' would translate
to the program 
`$\pif \;(x \neq y)\;\pthen\;\passert(\pfalse)\;\pelse\;\pskip$'.

The set of executions for programs with assertions
consists of sequences over the alphabet
$\Pi_\text{assert} = \Pi \cup \set{\dblqt{\passert (\pfalse)}}$
and can be defined as in \secref{program-execution},
with the addition of the following rule: $\exec(\passert(\pfalse)) = \passert (\pfalse)$

The functions $\val_\locT, \val_\dataT, \fldval_\locT$
and $\fldval_\dataT$ in the presence of $\dblqt{\passert (\pfalse)}$
are defined as before for execution prefixes without
$\dblqt{\passert (\pfalse)}$, and can be assumed to map
all elements in their domain to a special symbol $\bot$
for all executions containing $\dblqt{\passert (\pfalse)}$.
The feasibility of a partial execution on the alphabet
$\Pi_\text{assert} \setminus \set{\dblqt{\passert (\pfalse)}}$
on a heap structure $\Mm$ is defined as before (\secref{prog-semantics})
and is undefined otherwise.

\begin{definition}[Assertion Checking for Heap-Manipulating Programs]
The problem of assertion checking for a program $s \in \stmt_\text{assert}$
is to check whether for every heap structure $\Mm$ and for 
every partial execution of $s$ of the form
$\sigma \cdot \dblqt{\passert (\pfalse)}$, we have that $\sigma$ is infeasible
on $\Mm$.
\end{definition}

We first note that the above problem is undecidable in general, a direct consequence of~\cite{coherence2019}. Indeed, uninterpreted
programs are heap-manipulating programs
that do not mention any heap variables.
\begin{theorem}
Assertion checking for heap-manipulating programs is undecidable.
\end{theorem}

Finally, we state our first decidability result.
In the following, an \alaw~ coherent program is a program that
is both coherent and \alaw.

\begin{theorem}[Decidability of Uninterpreted Assertion Checking]
\thmlabel{alaw-coherent-assertion-decidable}
Assertion checking is decidable for the class of \alaw~\coherent~ programs and is $\pspc$-complete.
\end{theorem}

The proof of the above result relies on the following observations. First, for alias-aware executions, the terms associated with variables reflect their relationships on all heap structures (\thmref{terms-sanity}). Second, in this case, the streaming congruence closure algorithm for checking feasibility of coherent executions introduced in~\cite{coherence2019} can be extended directly to the case of coherent executions of heap-manipulating programs. The $\pspc$-hardness follows from the $\pspc$-hardness result for uninterpreted programs without updates (\cite{coherence2019}). As a consequence of these observations, we also obtain the following result:

\begin{theorem}
The problem of checking membership in the class of coherent and \alaw~ programs is decidable and is in $\pspc$.
\end{theorem}



\section{Memory Safety, Forest Data-Structures, and \Koherence}
\seclabel{forest}


We now tackle the problem of memory safety verification for heap-manipulating programs. 
The goal here is, given a program and an allocated set of locations (defined using a reachability specification), to check whether the program only dereferences (using pointer fields) locations that are in the (dynamically changing) set of allocated locations. We develop a technique for when the initial heap holds a \emph{forest data-structure} (disjoint lists and tree-like structures). We define a subclass of programs, called \koherent, for which memory safety is decidable. The key idea is to utilize the fact that the initial heap is a forest data-structure in order to make programs alias aware. 

For an execution that works over a forest data-structure,
one can accurately infer the aliasing relationship between two program variables by tracking
if they point to locations on the heap obtained by traversing the same path (i.e., starting from the same initial location and taking the same pointers at each step).
Intuitively, in a forest data-structure, two distinct traversals always lead to \emph{different} locations.
In addition, when the program execution does a single pass over the data-structure, 
one can keep track of the aliasing relationship between variables
(or equivalently, whether the locations pointed to by two variables are same or not)
using a \emph{streaming} algorithm.
The notion of \koherence\ essentially ensures such a single pass.
Finally, updatable maps can be used to keep track of the initial allocated set and the allocations/frees performed by the program during its execution. Consequently, memory safety can be reduced to assertion checking.
 

This section is organized as follows. ~\secref{reachability-specification} introduces the reachability specifications that specify the initial set of allocated locations, and defines the memory safety problem formally. We also show here that memory safety is undecidable in general, and undecidable even for coherent programs. 
~\secref{forest-datastructure-definition} defines the class of forest data-structures, and shows that memory safety remains undecidable for programs that start with a heap that holds a forest data-structure.~\secref{koherent-definition} defines the notion of \koherent~  programs. 



\subsection{Reachability Specification and Memory Safety}
\seclabel{reachability-specification}

\paragraph{Reachability Specification}
Heap-manipulating programs can be annotated by a 
\emph{reachability specification} that restricts the allowable nodes
that can be accessed by a program.
A reachability specification is an indexed set of 
triples $\reachspec = \set{\reachspec_k}_{k=1}^n$
where $\reachspec_i = (\startlocs_i, \reachptrs_i, \stoplocs_i)$
is such that
$\startlocs_k \subseteq \Cc_\locT$,
$\stoplocs_k \subseteq \Cc_\locT$
and $\reachptrs_k \subseteq \Ff_\locT$.
Each triple $\reachspec_i$ denotes a set of locations
$\reach_i$, which is the least set that contains $\startlocs_i$, does not include $\stoplocs_i$
and is closed under repeated applications of pointer fields $\reachptrs_i$.
More formally, given a heap structure $\Mm$ with interpretation $\Ii$,
$\reach_i$ gives a set of locations, which is the smallest set such that
\begin{enumerate*}[label=(\alph*)]
\item $\setpred{\Ii(c)}{c \in \startlocs_i} 
\setminus \setpred{\Ii(c)}{c \in \stoplocs_i} \subseteq \reach_i$, and 
\item for every $e \in \reach_i$ and for every $p \in \reachptrs_i$,
if $\Ii(p)(e) \not\in \setpred{\Ii(c)}{c \in \stoplocs_i}$, then
$\Ii(p)(e) \in \reach_i$.
\end{enumerate*}
We let $\reach_\reachspec = \bigcup\limits_{i=1}^n \reach_i$.
Often the heap structure $\Mm$ will be implicit and we will omit mentioning it.

\paragraph{Allocated Nodes}
Starting with a reachability specification $\reachspec$
on a given heap structure $\Mm$, an execution $\sigma$
defines a set of \emph{allocated nodes}, which we
denote as $\alloc(\sigma)$ and define as follows.
\begin{align*}
\begin{array}{rcl}
\alloc(\varepsilon) &=& \reach_\reachspec \\
\alloc(\sigma \cdot s) &=& 
\begin{cases}
\alloc(\sigma) \cup \set{\val_\locT(\sigma \cdot s, x)} & \text{ if } s = \dblqt{\palloc{x}}\\
\alloc(\sigma) \setminus \set{\val_\locT(\sigma, x)} & \text{ if } s = \dblqt{\pfree{x}}\\
\alloc(\sigma) & \text{otherwise}
\end{cases}
\end{array}
\end{align*}


\paragraph{Memory Safety}
An execution $\sigma$ is said to \emph{violate memory safety} over a
heap structure $\Mm$ with respect to a reachability specification $\reachspec$
if there is a prefix $\rho' = \rho \cdot s$ of $\sigma$
such that $\rho$ is feasible over $\Mm$
and one of the following holds.
\begin{enumerate}
    \item $s$ dereferences a location that was not allocated.
    That is, $s$ is of the form 
    $\dblqt{w \passign \pderef{y}{h}}$ or
    $\dblqt{\pderef{y}{h} \passign w}$, 
    $y \in V_\locT$ and $w$ and $h$ are variables and pointer fields 
    of appropriate sorts,
    such that
    $\val_\locT(\rho, y) \not\in \alloc(\rho)$.

    \item $s$ frees an unallocated location. That is, $s$ is of the form $\pfree{x}$ and $\val_\locT(\rho, x) \not\in \alloc(\rho)$.
\end{enumerate}
An execution $\sigma$ is \emph{memory safe} over $\Mm$ with respect to
$\reachspec$ if it does not violate memory safety over $\Mm$ with respect to $\reachspec$.
With this, we can now define the memory safety verification problem.
In the following, we fix our signature $\Sigma$.
\begin{definition}[Memory Safety Verification]
The memory safety verification problem asks, given a program $P \in \stmt$
and a reachability specification $\reachspec$, whether
for all heap structures $\Mm$, each execution $\sigma \in \exec(P)$ is memory safe over $\Mm$ with 
respect to $\reachspec$.
\end{definition}

We show, unsurprisingly, that checking memory safety is
undecidable in general.

\begin{theorem}[Undecidability of Memory Safety]
\thmlabel{memsafe_undec}
Memory safety verification is undecidable.
\end{theorem}

\begin{proof}
In~\cite{coherence2019}, the authors consider uninterpreted programs, which are programs that
have variables taking values in a data domain that is uninterpreted; programs in~\cite{coherence2019}
don't have heap variables, and do not modify heaps. It was shown (Theorem~11 in \cite{coherence2019}) that given
an uninterpreted program $P$, the problem of determining if there is a data domain $\Mm$ and an execution $\rho$ 
of $P$, such that $\rho$ is feasible in $\Mm$, is undecidable. Our result here can be proved by a simple reduction
from that problem. Let $P$ be an uninterpreted program. Consider the reachability specification $\varphi = 
(\{\init{x}\},\{p\},\{\init{x}\})$
such that $x \in V_\locT$ is a new variable. Consider program $P' = P; y = \pderef{x}{p}$. Observe that $P'$ is memory
safe with respect to $\varphi$ if and only if $P$ does not have a feasible execution with respect to some data model. 
\end{proof}



Given the undecidability result in~\thmref{memsafe_undec}, we need to identify a restricted subclass of programs and initial heap structures for which the problem of verifying memory safety is decidable. This leads us to the notions of forest data-structures and {\koherent} programs. In some sense, forest data-structures are natural classes of heap structures where the \alaw~ restrictions become very minimal. We combine these remaining restrictions with our usual notion of \coherence~ to introduce a new notion of \koherent~ programs. We define these below. 


\subsection{Forest Data-Structures}
\seclabel{forest-datastructure-definition}


Let us now define our characterization of heap structures that will be amenable to
our decidability result.
The main restriction is stated in bullet~\itemref{forest-main} below and
intuitively disallows two different paths to the same location on the heap.
\begin{definition}[Forest Data-structures]
\deflabel{forest}
A heap structure $\Mm = (U_\locT, U_\dataT, \Ii)$ over a 
signature $\Sigma$ is said to be a forest data-structure with respect 
to a reachability specification $\reachspec = \set{\reachspec_k}_{k=1}^n$ if
\begin{enumerate}
	\item\itemlabel{forest-singleton-stop} for every $1\leq i\leq n$, each set of stopping locations is a singleton set of the form $\stoplocs_i = \set{\stoploc_i}$,
	\item\itemlabel{forest-stop-pointers} for every $c \in \bigcup\limits_{k=1}^n\stoplocs_k$,
	and for every $f \in \Ff_\locT$, we have that $\Ii(f(c)) = \Ii(c)$,
     and
	\item\itemlabel{forest-main} for every
$t_i \in \Terms(\startlocs_i, \reachptrs_i) \cup \stoplocs_i$ and 
$t_j \in \Terms(\startlocs_j, \reachptrs_j) \cup \stoplocs_j$
we have, 
if $\Ii(t_i) = \Ii(t_j)$, then
either $t_i = t_j \in \startlocs_i \cap \startlocs_j$, or
$\Ii(t_i) = \Ii(t_j) = \Ii(\stoploc_i) = \Ii(\stoploc_j)$.
\end{enumerate}
\end{definition}
Intuitively, a heap structure is a forest data-structure
with respect to $\reachspec$, if the subgraph $G_i$ 
induced by the set of nodes (excluding $\stoplocs_i$) reachable from $\startlocs_i$,
using any number of pointers from $\reachptrs_i$ forms a tree,
and further, any two subgraphs $G_i$ and $G_j$ do not have a node in common (except possibly for the starting locations).
Notice that~\defref{forest} do not impose any restrictions on the elements of
the data sort of a heap structure.

Our notion of forest data-structures handles the aliasing 
problem while still being able to express many practical reachability specifications. Consider the following example similar to our pathological example from~\secref{examples-challenges}.

\begin{example} 
$\pi_2 = \pi'_2 \cdot \passume(\cd{z}_2 = \cd{z}_3)$ where
\begin{align*}
&\pi'_2 \delequal
\passume(\cd{x} \neq \nil)
\cdot
\passume(\cd{y} \neq \nil)
\cdot
\cd{z}_1 \passign \pderef{\cd{x}}{\cd{next}}
\cdot 
\passume(\cd{z}_1 \neq \cd{z}_2)
\cdot
\pderef{\cd{y}}{\cd{next}} \passign \cd{z}_2
\cdot \\
&\;\;\;\;\cd{z}_3 \passign \pderef{\cd{x}}{\cd{next}}
\end{align*}

With respect to forest data-structures the above pathological 
execution $\pi_2$ is, in fact, infeasible. 
To understand this, observe that
for a forest data-structure, if two different locations 
($\cd{x}$ and $\cd{y}$ in the execution $\pi_2$ for instance) 
are not equal to the stopping location 
($\cd{x}\neq\nil$ and $\cd{y}\neq\nil$), 
then these locations are also different from each other 
(and thus $\cd{x}\neq\cd{y}$). 
This means that in the execution $\pi_2$, the update 
$\dblqt{\pderef{\cd{y}}{\cd{next}} \passign \cd{z}_2}$ 
will not affect the value of $\dblqt{\pderef{\cd{x}}{\cd{next}}}$.
Now, when the execution reads the value of  
$\dblqt{\pderef{\cd{x}}{\cd{next}}}$ in $\cd{z}_3$, 
it is expected to be the same as earlier (i.e., some location as pointed to by $\cd{z}_1$), 
but since $\cd{z}_1 \neq \cd{z}_2$, 
we must have $\cd{z}_3 \neq \cd{z}_2$. 
This means that the last assume $\passume(\cd{z}_2 = \cd{z}_3)$ makes the execution infeasible.
\end{example}
 
However, the notion of forest data-structures, by itself, is not enough to ensure decidability, as is shown by the following result.

\begin{theorem}
\label{memsafe_undec_forest}
The memory safety verification problem for forest data-structures is undecidable.
\end{theorem}
\begin{proof}
Follows trivially from~\thmref{memsafe_undec} because the reachability specification $\reachspec$ used in the proof of~\thmref{memsafe_undec} is such that 
$\reach_\reachspec$ is the empty set.
\end{proof}

Next, we introduce a class of programs called \koherent\ programs that,
in conjunction with forest data-structures results in a $\pspc$ decision procedure.
We will discuss the decision procedure in~\secref{automaton}.


\subsection{\Koherent~Executions and Programs}
\seclabel{koherent-definition}

In this section, we identify the class of~\koherent~programs and executions,
for which we show decidability.
It is analogous to the notion of coherence (\defref{coherence}), except that it uses a different notion of equivalence (for terms) instead of $\rcong_{\alpha(\cdot)}$.
We define this new notion of equivalence (\emph{forest equality closure}) below.
Intuitively, this notion characterizes the behavior of executions on the class of forest data-structures
and allows us to adapt the algorithm used for assertion checking of \alaw~\coherent\ programs (\thmref{alaw-coherent-assertion-decidable}).
\begin{definition}[Forest Equality Closure]
Let $\reachspec = \set{\reachspec_i}_{i=1}^n$ be a reachability specification,
with $\reachspec_i = (\startlocs_i, \reachptrs_i, \stoplocs_i)$.
Let $\Terms_i = \Terms(\startlocs_i, \reachptrs_i)$ ($1\leq i\leq n$).
Let $E \subseteq \Terms \times \Terms$ be an equality relation on terms.
The \emph{forest equality closure} of $E$ with respect to $\reachspec$, 
denoted $\eqclosure(\reachspec, E) \subseteq \Terms \times \Terms$
is the smallest congruence relation that satisfies the following.
\begin{itemize}
	\item $E \subseteq \eqclosure(\reachspec, E)$.
	\item For every $t_i \in \Terms_i$ and $t_j \in \Terms_j$ such that
	$t_i \neq t_j$, we have
	\[
	(t_i, t_j) \in \eqclosure(\reachspec, E) \implies \set{(t_i, \stoploc_i), (t_j, \stoploc_j)} \subseteq \eqclosure(\reachspec, E)
	\]
\end{itemize}
\end{definition}

We will now define the notion of \emph{\koherent~ programs} with respect to a given reach specification $\reachspec$. To do this, we first need the notion of \koherent~ executions. This notion is very similar to the notion of \coherence\ (\defref{coherence}) and requires that terms are not recomputed and that assumes over the data sort are early.

\begin{definition}[\Koherence]
\deflabel{koherence}
Let $\sigma$ be a complete or partial execution. $\sigma$ is said to be {\koherent} if it satisfies the following two conditions.

\begin{description}
    \item[\bf Memoizing.] Let $\rho$ be a prefix of $\sigma$ of the form 
    $\rho'\cdot \dblqt{x \passign \pderef{y}{h}} $ or $\rho'\cdot \dblqt{x \passign f(\vec{y})} $ and let $t = \comp(\rho,x)$. 
    If there is a term $t' \in \Terms(\rho')$ such that $t \cong_{\eqclosure(\reachspec,\alpha(\rho'))} t'$, then it must be the case that there is some variable $z \in V$ such that $\comp(\rho',z) \cong_{\eqclosure(\reachspec,\alpha(\rho'))} t$.
    This condition is applicable to every sort-sensible combination of symbols.
    
    \item[\bf Early Assumes.] Let $\rho$ be a prefix of the form $\rho'\cdot \dblqt{\passume(u = v)}$ where $u,v \in V_\dataT$ and $t_u = \comp(\rho',u), t_v = \comp(\rho',v)$. If there is a term $t \in \Terms(\rho')$ such that $t$ is a superterm of $t_u$ or $t_v$ modulo $\cong_{\eqclosure(\reachspec,\alpha(\rho'))}$, then there must be some variable $w \in V_\dataT$ such that $\comp(\rho',w) \cong_{\eqclosure(\reachspec,\alpha(\rho'))} t$.
\end{description}
\end{definition}

Observe that the above definition is close the definition of \coherence\ (\defref{coherence}),
but nevertheless there are important differences.
First, the early assumes requirement does not apply to variables of the $\locT$ sort because forest data-structures simplify that requirement. 
Second, the congruence with respect to which equalities are demanded in the above definition is not one of the congruence closure defined by the equalities accumulated by the execution (which is the congruence used in~\defref{coherence}), but rather the congruence of forest equality closure.

\begin{definition}
A program is said to be~\koherent~ with respect to $\reachspec$ if all its executions are~\koherent~ with respect to $\reachspec$.
\end{definition}

We note here that \koherent~ programs on forest data-structures are essentially \emph{one-pass} algorithms. Intuitively, forest data-structures enforce the \alaw~ restriction by mandating that two locations obtained by two different traversals from the set of initial locations (therefore being represented by two different \emph{terms}) are different (with minor exceptions). Therefore if a \koherent~ execution computes a term twice, i.e., visits a location twice, by definition it has to store a pointer to that location in some variable. Since the number of variables is fixed \emph{a priori}, it is simple to see that a nontrivial, multi-pass algorithm such as linked-list sorting would inherently be non-streaming-coherent (since it is meant to work on lists of arbitrary size).

\section{Streaming Congruence Closure for Forest Data-structures}
\seclabel{automaton}

We will now focus on {\koherent} programs whose initial heap is a forest data-structure. We now present the second main result of our paper.


\begin{theorem}[Decidability of Memory Safety for \Koherent~Programs and Forest Data-Stru\-ctures]
\thmlabel{memory-safety-decidable}
The memory safety verification problem over forest data-structures
for $\koherent$ programs is decidable and is $\pspc$-complete.
\end{theorem}

Given a reachability specification, we present an algorithm for checking memory safety of such programs. This will establish the decidability and the membership of the problem in $\pspc$. The hardness follows from the $\pspc$-hardness result in~\cite{coherence2019}.

Our algorithm is automata theoretic--- we construct a finite state 
automaton $\Aa_\memsafe$ such that its language $L(\Aa_\memsafe)$ includes all {\koherent} executions that are memory safe and excludes all {\koherent} executions 
that are not memory safe.
Finally, in order to determine if a given \koherent\ program $P$ is memory safe, we simply check if $\exec(P) \subseteq L(\Aa_\memsafe)$; since $\exec(P)$ is also a regular language, this reduces to checking intersection of regular languages.

Let us first discuss some intuition for our construction. The state of the automaton has several components which can be categorized into two groups:
\begin{enumerate*} 
\item those that track feasibility of executions, and
\item those that are used to check for violations of memory safety.
\end{enumerate*}
The first category comprises of three components $\rcong, d, P$ which keep track of equalities, disequalities, and functional relationships between variables at a given point in the execution. These components are inspired by the work of~\cite{coherence2019} in the context of \emph{coherent uninterpreted programs}, where this information was used to answer questions of assertion checking/feasibility.
Our extended and refined notion of~\koherent~ programs also has this property, and these three components serve the same roles in our setting. 
Similar to the work in~\cite{coherence2019}, when an execution is \koherent, these three components $\rcong, d, P$ can be accurately maintained at every step in the execution by performing \emph{local congruence closure}.

We now describe the remaining components of the state, namely those that help check for memory safety.
We know from our earlier discussion that forest data-structures are implicitly alias aware. 
This is because such heap structures consist of many disjoint reachability specifications each of which describes a tree. Therefore, nodes obtained by different pointer traversals from the initial locations are different, and each tree in the forest is closed under a given set of pointers. We keep track of locations (for the purpose of memory safety) primarily using three ordered collections which are disjoint: 
\begin{enumerate*}[label=(\alph*)]
\item a collection $Y = \set{Y_i}_{i}$ of `yes' sets,
\item a collection $M = \set{M_i}_i$ of `maybe' sets and 
\item a `no' set $N$. 
\end{enumerate*}
Each $Y_i$ and $M_i$ above correspond to the $i^{th}$ reachability specification $\reachspec_i$.
In the following, we describe each of these.

$N$ is a set of variables that point to locations which, when dereferenced, lead to a memory safety violation. These locations include the stopping locations of any of the reachability specifications, or those locations in the allocated set that have been freed during the execution.
Consider the reachability specification $\psi = (\set{\cd{x}},\set{\cd{next}},\set{\nil})$. At the beginning of any execution, the set $N$ for this specification is $\set{\nil}$. Further, when the execution encounters the statement $\pfree{\cd{y}}$, we update the $N$ component of the state by adding $\cd{y}$.

The sets $Y_i$ hold a subset of program variables that point to locations obtained by traversing the $i^\text{th}$ reachability specification $\reachspec_i$. In particular, a variable $\cd{z}$ belongs to $Y_i$ when we know that the corresponding location can be dereferenced without causing a memory safety violation in any heap structure. This can happen when the execution establishes that the location pointed to by $\cd{z}$ is not the stopping location of $\reachspec_i$ (i.e., $\stoploc_i$). Consider $\psi$ and the single step execution $\sigma_1 = \passume(\cd{x} \neq \nil)$. At the end of $\sigma_1$ we know that in all heap structures in which $\sigma_1$ is feasible, it must be that the location pointed to by $\cd{x}$ can be dereferenced, and therefore we would add $\cd{x}$ to the corresponding `yes' set.

Lastly, the `maybe' sets $M_i$ hold variables that can be
obtained by traversing the tree of $\reachspec_i$, 
but are neither known to be stopping locations/freed
locations, nor known to be in the allocated set in all feasible heap structures. 
Consider again $\psi$ and the execution 
$\sigma_2 = \passume(\cd{x} \neq \nil) \cdot \cd{y} \passign \pderef{\cd{x}}{\cd{next}}$. 
Recall that $\cd{x}$ belongs to the `yes' set after 
the first statement (i.e., the execution $\sigma_1$) 
and we can therefore dereference it. 
However, at the end of the execution we would include $\cd{y}$ in the `maybe' set because the location $\pderef{\cd{x}}{\cd{next}}$ 
is not in the allocated set in all feasible heap structures.
In particular, a heap structure defining a list of length one is a feasible structure for this execution but $\pderef{\cd{x}}{\cd{next}}$ would be $\nil$, which is not in the allocated set.
Now consider the execution $\sigma_3 = \sigma_2 \cdot \passume(\cd{y} \neq \nil)$. After the last statement in $\sigma_3$ we would now shift $\cd{y}$ from the `maybe' set to the `yes' set. 

This is precisely what our automaton does, shifting variables between components and flagging memory safety violations when the execution dereferences (the location pointed to by) a variable that is not in any $Y_i$. However, we do this modulo equalities between variables and these sets described above are in fact sets of equivalence classes. Lastly, we also keep track of variables that point to locations allocated by $\dblqt{\palloc{\cdot}}$ using a set $A$ of equivalence classes of variables. Of course, we would transfer an equivalence class from $A$ to $N$ if the execution frees this memory, similar to what happens to members of the `yes' sets.

\vspace{-0.05in}
\subsection{Formal Description of Construction}

We shall now proceed with the formal construction of the automaton $\Aa_\memsafe$, which accepts a \koherent\ execution iff it is memory safe. 
Recall that our reachability specification is an indexed set of tuples
$\reachspec = \set{\reachspec_k}_{k=1}^n$, where
$\reachspec_k = (\startlocs_k, \reachptrs_k, \stoplocs_k)$.
To simplify presentation, we
assume that the set of variables $V$ in our programs is 
such that for every constant $c$ appearing in the reachability 
specification $\reachspec$, there is a variable $v_c$ 
corresponding to $c$.
Further, we assume that these variables are never overwritten.
We will, therefore, often interchangeably refer to these
constants in $\reachspec$ by their corresponding variables, and vice versa.
These assumptions can be relaxed with a
more involved construction.

The automaton is a tuple $\Aa_\memsafe = (Q, q_0, \delta)$, where $Q$ is the set of states, $q_0$ is the initial state and $\delta : Q \times \Pi \to Q$ is the transition function. 
Recall that executions are strings over $\Pi$, which is also the alphabet of the automaton $\Aa_\memsafe$. 
We describe each of these components below.

\textbf{States.}
The automaton has two distinguished states $\infeasible$
and $\unsafe$. All other states are tuples of the form
$(\rcong, d, P, Y, M, N, A, \xtra)$, where each component is as follows.
\begin{itemize}
	\item $\rcong$ is an equivalence relation over $V$
	that respects sorts. We will use $\eqcl{x}{\rcong}$ to denote the set $\setpred{y}{(x, y) \in \rcong}$
	\item $d$ is a symmetric set of pairs of the form $(c_1, c_2)$,
	where $c_1, c_2 \in V / \rcong$ are equivalence classes.
	\item $P$ associates partial mappings to function symbols in $\Ff_\dataT$
	and pointer field in $\fields = \fields_\locT \cup \fields_\dataT$. 
	More formally, for every $f \in \Ff_\dataT$ of arity $r$, 
	and classes $c_1, c_2, \ldots, c_r \in V_\dataT /\rcong$,
	we have
	$P(f)(c_1, \ldots, c_r) \in V_\dataT/\rcong $ (if defined).
	Similarly, for every $p \in \fields_\locT$, and every 
	$c \in V_\locT /\rcong$,
	$P(p)(c) \in V_\locT  /\rcong $, and
	for every $d \in \fields_\dataT$, and every 
	$c \in V_\locT /\rcong$,
	$P(d)(c) \in V_\dataT /\rcong$. 
	We will say $P(h)(c_1, \ldots, c_k) = \undf$ when the
	function/pointer symbol $h$ is not defined on the arguments $c_1, \ldots, c_k$.
	\item $Y = \set{Y_k}_{k=1}^n$ and $M = \set{M_k}_{k=1}^n$ are 
	indexed collections of sets
	such that for every $1\leq i\leq n$, we have
	$Y_i, M_i \subseteq V_\locT/\rcong$.
	\item The sets $N$, $A$ and $\xtra$ are sets of equivalence classes of location
	variables, i.e., $N, A, \xtra \subseteq V_\locT/\rcong$.
\end{itemize}

\textbf{Initial State.}
The initial state $q_0$ is the tuple $(\rcong_0, d_0, P_0, Y_0, M_0, N_0, A_0, \xtra_0)$
such that 
\begin{itemize}
\item $\rcong_0$ is the identity relation on the set $V$ of variables,
\item $P_0$ is such that for all functions and pointer fields $f$, 
the range of $P(f)$ is empty,
\item each of $d_0, A_0$ and $(Y_0)_1, \ldots, (Y_0)_n$ are $\emptyset$ (empty set),
\item for each $1\leq i\leq n$, $(M_0)_i = \setpred{\eqcl{v}{\rcong_0}}{v \in \bigcup_{i=1}^n \startlocs_i}$,
\item $N_0 = \setpred{\eqcl{v}{\rcong_0}}{v \in \bigcup_{i=1}^n \stoplocs_i}$.
\item $\xtra_0 = \setpred{\eqcl{v}{\rcong_0}}{v \in V_\locT} \setminus \left(N_0 \cup\bigcup_{i=1}^n (M_0)_i\right) $.
\end{itemize}

\textbf{Transitions}.
The states $\infeasible$ and $\unsafe$ are absorbing states.
That is, for every $s \in \Pi$, $\delta(\infeasible, s) = \infeasible$
and $\delta(\unsafe, s) = \unsafe$.
In the following, we describe the transition function for
every other state in $Q$.
Let $q \in Q \setminus \set{\infeasible, \unsafe}$, $s \in \Pi$
and let $q' = \delta(q, s)$.
Then, $q'$ is $\infeasible$ or $\unsafe$, or is of the form
 $q' = (\rcong', d', P', Y', M', N', A', \xtra')$.
\ucomment{In many of the transitions below, 
$\rcong'$ (resp. $d'$) will be obtained by
adding or removing pairs from $\rcong$ (resp. $d$).
For ease of notation, we will not explicitly ensure that $\rcong'$ (resp. $d'$) be an equivalence relation (resp. symmetric relation respecting $\rcong'$), but the reader should nevertheless assume $\rcong'$ (resp. $d'$) indeed satisfies this property.}

\begin{enumerate}
	\item \textbf{Case $s = \dblqt{u \passign v}$, $u, v \in V$}. \\
	In this case, we add the variable $u$ into the class of $v$ and
	appropriately update each of the components. That is,
	$\rcong' = \big(\rcong \setminus \setpred{(u, u'),(u', u)}{u\neq u', u' \in \eqcl{u}{\rcong}}\big) \cup \setpred{(u, v'), (v', u)}{v' \in \eqcl{v}{\rcong}}$.
	The other components of the state are the same as in $q$. 

	\item \textbf{Case $s = \dblqt{x \passign \pderef{y}{p}}$, $x, y \in V_\locT$ and $p \in \fields_\locT$}.\\
	\newcommand{\spindex}{k}
	In this case, we need to check if the variable $y$ corresponds to a location
	that can be dereferenced. If not, we have a memory safety violation;
	otherwise, we establish the relationship $x = p(y)$ in the next state.
	Formally, if there is no $i$ such that $\eqcl{y}{\rcong} \in Y_i$ and if $\eqcl{y}{\rcong} \not\in A$, then $q' = \unsafe$.
	Otherwise we have that $\eqcl{y}{\rcong} \in A$ or  there is a $\spindex$ such that $\eqcl{y}{\rcong} \in Y_\spindex$. 
	In this case we define the tuple $q' = (\rcong', d', P', Y', M', N', A', \xtra')$
	below.
	Here, we need to consider the following cases. 
	\begin{itemize}
		\item Case $P(p)(\eqcl{y}{\rcong})$ is defined and equals $\eqcl{z}{\rcong}$.
		In this case, $q'$ is defined in the same manner as though $s = \dblqt{x \passign z}$.

		\item Case $P(p)(\eqcl{y}{\rcong}) = \undf$. 
		Here, for $\koherent$ programs it must be the case that $\eqcl{y}{\rcong} \in Y_\spindex$ for some $\spindex$.
		Here, we create a new singleton equivalence class containing $x$
		and set the value of the $p$ map, on $y$ to be this new class.
		We also assert that $\eqcl{x}{\rcong'}$ is not equal to any class in any of the $Y_i$s or in $A$.
		That is, 
		\begin{itemize}
			\item $\rcong' = \rcong \setminus \setpred{(x, x')}{x' \neq x} \cup \set{(x, x)}$.		
			\item $d'{=} \setpredsmall{(\eqcl{u}{\rcong'}, \eqcl{v}{\rcong'})}{u \neq x, v \neq x,  (\eqcl{u}{\rcong}, \eqcl{v}{\rcong}) \in d} \cup \setpredsmall{(\eqcl{x}{\rcong'}, c), (c, \eqcl{x}{\rcong'})}{c \in A \cup \bigcup\limits_{i=1}^n Y_i}$
			\item $P'(p)(\eqcl{y}{\rcong'}) = \eqcl{x}{\rcong'}$. For all other combinations of functions/pointers and arguments, $P'$ behaves same as $P$.
			\item 
			One of the sets $M_\spindex$ and $\xtra$ are updated depending upon the pointer $p$. 
			If $p \in \reachptrs_k$, then
			$M'_\spindex = \setpredsmall{\eqcl{z}{\rcong'}}{ \eqcl{z}{\rcong} \in M_\spindex} \cup \set{\eqcl{x}{\rcong'}}$.
			Otherwise, $\xtra = \setpredsmall{\eqcl{z}{\rcong'}}{ \eqcl{z}{\rcong} \in \xtra} \cup \set{\eqcl{x}{\rcong'}}$.
			\item All other components are the same as in $q$.
		\end{itemize}
	\end{itemize}

	\item \textbf{Case $s = \dblqt{a \passign \pderef{y}{d}}$, $a \in V_\dataT, y \in V_\locT$ and $d \in \fields_\dataT$}.\\
	As in the previous case, $q' = \unsafe$ if $\eqcl{y}{\rcong} \not\in \cup_{i=1}^nY_i \cup A$.
	Otherwise, similar to the previous case, we have two cases to consider.
	As before, if there is a variable $b \in V_\dataT$ such that $P(d)(\eqcl{y}{\rcong}) = \eqcl{b}{\rcong}$, then we treat this case as that of $\dblqt{a \passign b}$.
	Otherwise, the new equivalence relation $\rcong'$ is such that
	$\rcong' = \rcong \setminus \setpred{(a, u)}{u \neq a} \cup \set{(a, a)}$,
	while the other components are the same as in $q.$

	\item \textbf{Case $s = \dblqt{\pderef{y}{h} \passign u}$, $y \in V_\locT$}.\\
	We uniformly handle the case of $h$ being either a pointer field ($\fields_\locT$) or a data field ($\fields_\dataT$).
	Here we have $q' = \unsafe$ if $\eqcl{y}{\rcong} \not\in \cup_{i=1}^nY_i \cup A$.
	Otherwise, we simply change $P$ as follows (while keeping other components the same as in $q$):
	\begin{itemize}
	    \item $P'(f) = P(f)$ for $f\neq h$
	    \item $P'(h)(\eqcl{y}{\rcong'}) = \eqcl{u}{\rcong'}$
	    \item For $z \in V_\locT$ such that $z \notin \eqcl{y}{\rcong}$, $P'(h)(\eqcl{z}{\rcong'}) = \eqcl{w}{\rcong'}$ if $P(h)(\eqcl{z}{\rcong}) = \eqcl{w}{\rcong}$  for some variable $w$, and 
	    $P'(h)(\eqcl{z}{\rcong'}) = \undf$ if no such variable $w$ exists.
	\end{itemize}

	\item \textbf{Case $s = \dblqt{a \passign f(c_1, \ldots, c_r)}$, $a \in V_\dataT$}.\\
	Here, we have two cases to consider again.
	If there is a variable $b \in V_\dataT$ such that $P(f)(\eqcl{c_1}{\rcong}, \ldots, \eqcl{c_r}{\rcong}) = \eqcl{b}{\rcong}$, then we treat this case as that of $\dblqt{a \passign b}$.
	Otherwise, we add a singleton equivalence class containing $a$ and update $P(f)$, while keeping all other components the same (modulo the new equivalence relation).
	Formally,
	\begin{itemize}
			\item $\rcong' = \rcong \setminus \setpred{(a, u), (u,a)}{u \in V_\dataT} \cup \set{(a, a)}$.		
			\item $P'(h)$ is same as in $q$ if $h\neq f$. The evaluation of $P'$ on $f$ is described as follows. 
			\begin{align*}
				P'(f)(\eqcl{u_1}{\rcong'}, \ldots, \eqcl{u_r}{\rcong'}) = 
				\begin{cases}
				\eqcl{a}{\rcong'} & \text{if for every $1\leq i\leq r$}, \eqcl{u_i}{\rcong} = \eqcl{c_i}{\rcong} \text{ and } a \neq u_i \\
				\eqcl{u}{\rcong'} & 
				\begin{aligned}
				&\text{otherwise if } a \not\in \set{u, u_1, \ldots, u_r} \\
				&\text{and } \eqcl{u}{\rcong} = P(f)(\eqcl{u_1}{\rcong}, \ldots, \eqcl{u_r}{\rcong})
				\end{aligned}\\
				\undf & \text{otherwise}
				\end{cases}
			\end{align*}
			\item All other components are the same as in $q$.
		\end{itemize}

	\item \textbf{Case $s = \dblqt{\palloc{x}}$}.\\
	In this case, we create a new singleton class containing $x$, 
	and add this class to the component $A$.
	We also assert that this new class is not equal to any other class.
	Formally, 
	\begin{itemize}
		\item $\rcong' = \rcong \setminus \setpred{(x, u)}{u \in V_\locT} \cup \set{(x, x)}$.
		\item $d' = \setpred{(\eqcl{u_1}{\rcong'},\eqcl{u_2}{\rcong'})}{u_1 \neq u_2, (\eqcl{u_1}{\rcong},\eqcl{u_2}{\rcong}) \in d \text{ or } u_1 = x \lor u_2 = x}$
		\item $A' = \setpred{\eqcl{u}{\rcong'} }{\eqcl{u}{\rcong} \in A} \cup \set{\eqcl{x}{\rcong'}}$.
		\item The component $P$ is updated as follows.
		For every function symbol $f \in \Ff_\dataT$, $P'(f)$ is the same as $P(f)$. Further, for every pointer field $h \in \fields_\dataT \cup \fields_\locT$ and for every variable $y \in V_\locT \setminus \set{x}$, we have $P'(h)(\eqcl{y}{\rcong'}) = P(h)(\eqcl{y}{\rcong})$.
		Finally, the mapping on $x$ is defined as $P(p)(\eqcl{x}{\rcong'}) = \eqcl{x}{\rcong'}$ for every location pointer $p \in \fields_\locT$ and $P(d)(\eqcl{x}{\rcong'}) = \undf$ for every data pointer $d \in \fields_\dataT$.
		\item All other components are updated as usual.
	\end{itemize}

	\item \textbf{Case $s = \dblqt{\pfree{x}}$}.\\
	In this case, if 
	$\eqcl{x}{\rcong} \not\in A \cup \bigcup\limits_{i=1}^n Y_i$,
	then $q' = \unsafe$.
	Otherwise, we remove the class $\eqcl{x}{\rcong}$ from the sets $A$ or $Y_1, \ldots, Y_n$ and add it to the set $N$.
	That is,
	\begin{itemize}
	    \item $\rcong' = \rcong$ 
	    \item $N' = \setpred{\eqcl{z}{\rcong'}}{\eqcl{z}{\rcong} \in N} \cup \set{\eqcl{x}{\rcong'}}$
	    \item for every $1\leq i\leq n$, $Y'_i = \setpred{\eqcl{z}{\rcong'}}{\eqcl{z}{\rcong} \neq \eqcl{x}{\rcong}, \eqcl{z}{\rcong} \in Y_i}$
	    \item $A = \setpred{\eqcl{z}{\rcong'}}{\eqcl{z}{\rcong} \neq \eqcl{x}{\rcong}, \eqcl{z}{\rcong} \in A}$
	    \item Other components remain the same.
	\end{itemize}

	\item \textbf{Case $s = \dblqt{\passume(x = y)}$, $x, y \in V_\locT$}.
	In this case, if $\eqcl{x}{\rcong} = \eqcl{y}{\rcong}$ then $q'= q$. Otherwise we have several cases to consider.
	
	In each of these cases, we construct a new tuple $q'' = (\rcong'', d'', P'', Y'', N'', M'', A'', \xtra'')$. 
	Finally, we set $q' = q''$ if $d'' \cap \rcong'' = \emptyset$; otherwise we have $q' = \infeasible$.
	    \begin{itemize}
	        \item The first case to consider is when 
	        $\set{\eqcl{x}{\rcong}, \eqcl{y}{\rcong}} \cap \bigcup\limits_{i=1}^n M_i= \emptyset$.
	        In this case, we merge $\eqcl{x}{\rcong}$ and $\eqcl{y}{\rcong}$. 
	        More formally, $\rcong''$ is the smallest equivalence relation such that $\big(\rcong \cup \set{(x,y)}\big) \subseteq \rcong''$. 
	        Further, for every component $C \in \set{\xtra, A, N} \cup \bigcup\limits_{i=1}^n \set{Y_i, M_i}$ with the corresponding component in $q''$ being $C''$, and for every $z \in V_\locT$ 
	        such that $z \notin \eqcl{x}{\rcong} \cup \eqcl{y}{\rcong}$, we have
	       $\eqcl{z}{\rcong''} \in C''$ iff $\eqcl{z}{\rcong} \in C$. 
	       For a variable $z \in \eqcl{x}{\rcong} \cup \eqcl{y}{\rcong}$, we have $\eqcl{z}{\rcong''} \in C''$ iff $\set{\eqcl{x}{\rcong}, \eqcl{y}{\rcong}} \cap C \neq \emptyset$.
	        The other components of $q''$ are the same as in $q$ (modulo the new equivalence classes).
	        
	        \item Otherwise, consider the case when $\eqcl{x}{\rcong} \in M_i$ for some $i$. In this case, in addition to adding $\set{(x,y)}$ we also add the pair $\set{(x,\stoploc_i)}$. Similarly if $\eqcl{y}{\rcong} \in M_j$ for some $j$ we add $\set{(y,\stoploc_j)}$. Construct the state $q''$ with $\rcong''$ being the smallest equivalence relation including these new pairs (and other components remaining the same).  
	    \end{itemize}
	\item \textbf{Case $s = \dblqt{\passume(x \neq y)}$, $x, y \in V_\locT$}.
	Similarly as above, in this case when $(\eqcl{x}{\rcong},\eqcl{y}{\rcong}) \in d$ we have $q' = q$. If $\eqcl{x}{\rcong} = \eqcl{y}{\rcong}$ then $q'  = \infeasible$. Otherwise, we have the following cases:
	\begin{itemize}
	    \item $\eqcl{x}{\rcong} = \eqcl{\stoploc_i}{\rcong}$ and $\eqcl{y}{\rcong} \in M_i$ for some $i$ (or vice versa). In this case, we simply put the equivalence class of $y$ into $Y_i$ and assert that $\eqcl{y}{\rcong}$ is unequal to all other classes.
	    More formally, $\rcong' = \rcong$, $Y'_i = Y_i \cup \set{\eqcl{y}{\rcong}}$ $d' = d \cup  \setpred{(\eqcl{y}{\rcong},\eqcl{z}{\rcong}), (\eqcl{z}{\rcong},\eqcl{y}{\rcong})}{z \notin \eqcl{y}{\rcong}}$. The other components remain the same.
	    \item Otherwise, we simply update $d' = d \cup \set{(\eqcl{x}{\rcong},\eqcl{y}{\rcong})}$; all other components are the same. 
	\end{itemize}
	\item \textbf{Case $s = \dblqt{\passume(a = b)}$, $a, b \in V_\dataT$}.
	In this case, we merge equivalence classes repeatedly and perform a `local congruence closure'. We construct a state $q''$ to determine if the transition must be to $\infeasible$. More formally, we define the state $q''$ with the $\rcong''$ component as the smallest equivalence relation such that:
	\begin{enumerate*}
	    \item $\rcong \cup {(a,b)} \subseteq \rcong''$
	    \item If $(u_i,v_i) \in \rcong$ for $1 \leq i \leq r$ and $\eqcl{w}{\rcong''} = f(\eqcl{u_1}{\rcong''}\ldots,\eqcl{u_r}{\rcong''}),\eqcl{w'}{\rcong''} = f(\eqcl{v_1}{\rcong''}\ldots,\eqcl{v_r}{\rcong''}) $ then $(w,w') \in \rcong''$.
	\end{enumerate*}
	
	The other components remain the same. In particular, it is correct to retain the $P$ component since the above construction is a congruence relation.
	Finally, if there exist $u,v \in V_\dataT$ such that $(u,v) \in \rcong''$ and $(\eqcl{u}{\rcong''},\eqcl{v}{\rcong''}) \in d''$ then $q' = \infeasible$. Otherwise, $q' = q''$.
	\item \textbf{Case $s = \dblqt{\passume(a \neq b)}$, $a, b \in V_\dataT$}.
	Similarly as above, if $\eqcl{a}{\rcong} = \eqcl{b}{\rcong}$ then $q' = \infeasible$. Otherwise, we update $d' = d \cup \set{(\eqcl{a}{\rcong},\eqcl{b}{\rcong})}$ and all other components remain the same.
\end{enumerate}

The following theorem states the correctness of the automaton $\Aa_\memsafe$.

\begin{lemma}
\lemlabel{automaton-correctness}
Let $\sigma$ be a $\koherent$ execution and 
let $\reachspec$ be a reachability specification
and let $q$ be the state of the automaton $\Aa_\memsafe$
after reading $\sigma$.
Then, $q = \unsafe$ iff
there is a forest data-structure $\Mm$ (with respect to $\reachspec$) such that 
$\sigma$ violates memory safety on $\Mm$.
\end{lemma}

The problem of checking if a \koherent\ program is memory safe against
a given specification is decided as follows.
Recall that the set of executions of a given program $P$ 
constitutes a regular language $\exec(P)$. 
Let $L(\Aa_\memsafe)$ be the set of executions $\sigma \in \Pi^*$
that do not go to the state $\unsafe$.
Then, the problem of checking if $P$ is memory safe
reduces to checking if $\exec(P) \subseteq L(\Aa_\memsafe)$.

Next, we show that the problem of checking \koherence~
is also decidable.
To address the problem of checking \koherence,
we construct an automaton $\Aa_\checkkoherence$ similar to $\Aa_\memsafe$
(similar to the automaton for checking \coherence~ in~\cite{coherence2019}) that keeps track
of the following information.
For every function/pointer $f$ of arity $r$, and for every tuple
$(x_1, \ldots, x_r)$ of variables (of appropriate sorts),
each state of the automaton $\Aa_\checkkoherence$
maintains a boolean predicate denoting whether or not
$f(x_1, \ldots, x_r)$ has been computed in any execution
that reaches the state. 
This gives us our next result.

\begin{theorem}
\thmlabel{caching-decidable}
The problem of checking whether a given program is \koherent\ with respect to a given reach specification defining a forest data-structure is decidable in $\pspc$.
\end{theorem}


\section{Implementation and Evaluation}
\seclabel{experiments}

We implemented a tool~\cite{streamveriftool} for deciding memory safety of forest data-structures based on the construction from~\secref{automaton}. The tool is \textasciitilde 2000 lines of Ocaml 4.07.0 code. It takes as input a program from the grammar presented in~\secref{program-syntax} annotated with a reachability specification, as in~\secref{reachability-specification}. 
The tool does not explicitly construct the automaton from~\secref{automaton}.
Instead, it implements a fixpoint algorithm, described below,
which determines the set of states associated with every program point, and
uses this set to check for memory safety.

\subsection{Fixpoint Algorithm} 
Here we give a high-level description of the implementation.
First, observe that $\Aa_\memsafe$ has exponentially many states in the number of program variables.
We manage this by implementing the transition $\delta$ from \secref{automaton} and only building automaton states as they are encountered.
For straight-line programs, each transition results in a single new state. For complex programs that use $\pif$-$\pthen$-$\pelse$ and $\pwhile$ statements, we need to keep track of a bag of states. To see this, suppose we are checking the program $\pif \, (c) \, \pthen \, s_1 \, \pelse \, s_2$. To begin our bag of states contains only the initial state $q_0$. In order to process the $\pif$-$\pthen$-$\pelse$, the initial state needs to make two separate transitions, one for each of the two executions generated by the two branches. The bag of states thus grows to include $s_\textsf{then} = \delta(q_0, \dblqt{\passume(c)})$ and $s_\textsf{else} = \delta(q_0, \dblqt{\passume(\neg c)})$. The branches can then take transitions starting from $s_\textsf{then}$ and $s_\textsf{else}$. The union of the reachable states from each branch gives us a new bag of states from which to process the remaining program. The intuition for $\pif$-$\pthen$-$\pelse$ carries over to $\pwhile$. From any state in our bag at the beginning of a $\pwhile$, we collect the bag of states that results from any number of executions of the loop guard and body, starting from that state. The number of states $|\Aa_\memsafe|$ is finite, and thus a fixed point is guaranteed. 
For the benchmarks considered here, the number of states explored by the tool is significantly smaller than the worst case.

If at any point the tool detects a memory safety violation it halts and reports the error. In addition to memory safety, it also monitors the \koherence~property as it processes the input program. To do so, it keeps track of terms that were computed using existing equivalence classes, but which were subsequently dropped. If the program attempts to compute the same term using the same classes, the implementation flags a failure of \koherence~and halts. For example, to process $\dblqt{a \passign f(c)}$, the tool checks that $f$ was not previously applied to $\eqcl{c}{\rcong}$ and later dropped. In general, this information can be maintained by remembering the equivalence classes $\eqcl{c_1}{\rcong}, \ldots, \eqcl{c_r}{\rcong}$ on which any function $f$ (of arity $r$) has been computed.

The algorithm implemented by the tool is more general than we have described thus far.
It can output the set of states reached at the head of a loop. By interpreting individual states as conjunctions of equality and disequality, and the set of states as a disjunction of such conjuncts, we can obtain an inductive invariant that proves memory safety (when the program is memory safe).
Any assertion in the form of a Boolean combination of equality statements on program variables can also be checked. This can be accomplished by appending the negated assertion to the end of the program and checking that all reachable states are infeasible. 

\subsection{Benchmarks}

We seek to answer the following basic questions about \koherent\ programs and our algorithm. First, is it the case that the most natural way to write heap-manipulating single-pass programs on lists and trees results in \koherence? Second, for \koherent~programs with and without memory safety violations, is the algorithm able to verify memory safety or find violations of it in the corresponding uninterpreted program? Third, how fast is the algorithm? Note that since we do abstract the primitive types and functions/relations on real programs, it is not clear that the tool will be able to prove memory safe programs as so. 

To answer the first and second questions, we wrote natural heap-manipulating programs over singly-linked lists (sorted and unsorted) and tree data-structures (binary search trees, AVL trees, rotations of trees) in our input language, and evaluated the tool on them to determine if they were \koherent and to test for memory safety. 

The first column of Table~\ref{tab:evaluation} gives
the set of programs in our benchmark.
These are typically \emph{single-pass} algorithms over an input data-structure. For example, finding a key in a binary search tree or in-place reversal of a linked list are single pass algorithms. 



\begin{table}
  \centering
    \caption{Evaluation of tool for proving memory safety and finding
  memory safety errors}
  \scalebox{0.85}{
  \begin{tabular}{l|r|c|c|r|c}
    \hline
    \hline
    Program & LOC & Streaming-~ & Found & \# States & Time \\
    & & coherent? & Safe &  & (ms)     \\ 
     \hline
{\bf Verification of Memory Safe Programs} & & & & & \\
\hline
\benchfont{sll-append-safe} & 19 & yes & \ding{51} & 4 & 4 \\
\benchfont{sll-copy-all-safe} & 27 & yes & \ding{51} & 6 & 3 \\
\benchfont{sll-delete-all-safe} & 56 & yes & \ding{51} & 58 & 9 \\
\benchfont{sll-deletebetween-safe} & 42 & yes & \ding{51} & 53 & 8 \\
\benchfont{sll-find-safe} & 16 & yes & \ding{51} & 4 & 3 \\
\benchfont{sll-insert-back-safe} & 20 & yes & \ding{51} & 3 & 3 \\
\benchfont{sll-insert-front-safe} & 8 & yes & \ding{51} & 1 & 3 \\
\benchfont{sll-insert-safe} & 50 & yes & \ding{51} & 12 & 4 \\
\benchfont{sll-reverse-safe} & 12 & yes & \ding{51} & 3 & 3 \\
\benchfont{sll-sorted-concat-safe} & 17 & yes & \ding{51} & 4 & 3 \\
\benchfont{sll-sorted-insert-safe} & 50 & yes & \ding{51} & 12 & 4 \\
\benchfont{sll-sorted-merge-safe} & 74 & yes & \ding{51} & 62 & 8 \\
\benchfont{bst-find-safe} & 23 & yes & \ding{51} & 21 & 4 \\
\benchfont{bst-insert-safe} & 45 & yes & \ding{51} & 29 & 6 \\
\benchfont{bst-remove-root-safe} & 52 & yes & \ding{51} & 12 & 4 \\
\benchfont{avl-balance-safe} & 190 & yes & \ding{51} & 48 & 12 \\
\benchfont{tree-rotate-left-safe} & 25 & yes & \ding{51} & 3 & 3 \\
    \\ \hline
{\bf Finding Errors in Memory-unsafe Programs} & & & & & \\
\hline
sll-append-unsafe & 20 & yes & \ding{55} & --- & 3 \\
sll-copy-all-unsafe & 29 & yes & \ding{55} & --- & 4 \\
sll-delete-all-unsafe & 58 & yes & \ding{55} & --- & 5 \\
sll-deletebetween-unsafe & 44 & yes & \ding{55} & --- & 4 \\
sll-find-unsafe & 18 & yes & \ding{55} & --- & 3 \\
sll-insert-back-unsafe & 20 & yes & \ding{55} & --- & 3 \\
sll-insert-front-unsafe & 9 & yes & \ding{55} & --- & 3 \\
sll-insert-unsafe & 50 & yes & \ding{55} & --- & 3 \\
sll-reverse-unsafe & 12 & yes & \ding{55} & --- & 3 \\
sll-sorted-concat-unsafe & 17 & yes & \ding{55} & --- & 3 \\
sll-sorted-insert-unsafe & 50 & yes & \ding{55} & --- & 3 \\
sll-sorted-merge-unsafe-1 & 69 & yes & \ding{55} & --- & 4 \\
sll-sorted-merge-unsafe-2 & 63 & yes & \ding{55} & --- & 3 \\
bst-find-unsafe & 25 & yes & \ding{55} & --- & 4 \\
bst-insert-unsafe & 49 & yes & \ding{55} & --- & 6 \\
bst-remove-root-unsafe & 54 & yes & \ding{55} & --- & 3 \\
avl-balance-unsafe & 111 & yes & \ding{55} & --- & 3 \\
tree-rotate-left-unsafe & 20 & yes & \ding{55} & --- & 3 \\
    \\ \hline
{\bf Detecting Programs are \emph{not} } & & & & & \\
{\bf \koherent} & & & & & \\
\hline
sll-sorted-merge-non-\koherent & 75 & no & --- & --- & 4 \\
bst-remove-root-non-\koherent & 55 & no & --- & --- & 3 \\
    \hline
        \hline
        \vspace{0.1cm}
  \end{tabular}
  }
  \label{tab:evaluation}
\end{table}

The names of the programs indicate whether
or not the program truly contains an unsafe memory access (i.e., the ground truth).
Programs whose names end in `\benchfont{unsafe}' were obtained by introducing one of two possible memory safety errors into their `\benchfont{safe}' counterparts:
(i) attempts to read or write to a location that is unallocated, and (ii) freeing unallocated memory locations. 

One example of the first kind is illustrated in \benchfont{sll-copy-all}, which copies the contents of a linked list into a freshly allocated list. In this example, the program steps through the input list in a loop until it reaches $\nil$. In each iteration, a new node is allocated, initialized with the contents of the current node, and connected to the end of the new list. The program relies on the invariant that the new list has a next node to step to whenever the old list does. Thus, it does not perform a $\nil$ check when advancing along the next pointer for the new list. The \benchfont{sll-copy-all-\textit{unsafe}} fails to maintain the invariant by incorrectly adding the freshly allocated node to the new list. An example for errors of the second kind (freeing memory locations that may not be allocated) can be found in 
\benchfont{sll-deletebetween-unsafe}. 
In this example, the task is to delete all nodes 
in a linked list that have key values in a certain range. 
The mistake in this example happens when the program has found a node to delete, but, instead of saving the next node and deleting the current node, it instead frees the next node, which may be unallocated.

\subsection{Discussion of Results}
Table~\ref{tab:evaluation} shows the results for the evaluation, which was performed on a machine running Ubuntu 18.04 with an Intel i7 processor at 2.6 GHz. 
Columns 3-6 pertain to the operation of the algorithm on the benchmarks. Column 3 indicates whether or not the benchmark fails the \koherence~co\-ndition. Our tool terminates and identifies memory safety and violation of memory safety on all \koherent\ programs. Column 4 depicts whether or not an unsafe memory access was detected. Column 5 gives the total number of states that are reachable at the end of the program. Note that non-\koherence~and violation of memory safety preclude each other in the table. Upon detecting either, the algorithm halts (and we do not report the number of reachable states). Column 6 gives the total running time of the tool on each benchmark, which is negligible in all cases. Note that the number of reachable states for each example is also quite small relative to the total number of possible states, which grows faster than the Bell numbers. That our algorithm only examines a small fraction of the total state space is encouraging, and suggests that it may scale well for much larger and more complex programs.


\section{Related Work}
\seclabel{related}

Memory safety errors have attracted a lot of attention because they are serious security vulnerabilities that have been exploited to attack systems~\cite{nagarakatte2015,Szekeres2013};
they are one of the most common causes of security bugs~\cite{slashdotpost}.
Memory safety concerns have even led to new programming languages, such as Rust~\cite{rustlangdoc}, that statically assure memory safety (while being efficient).
Memory safety vulnerabilities of programs written in C/C++ are still of great concern,
and, consequently, identifying fundamental techniques that
establish decidability of the problem even for restricted classes of programs is interesting.

There is a rich literature of 
preventing memory safety errors at runtime by instrumenting code with runtime checks, or at compile time~(see \cite{nagarakatte2015}
and references therein, SafeC~\cite{safeclib}, 
CCured~\cite{necula2002ccured,AustinEtAlPLDI94,NeculaTOPLAS05,ConditPLDI03}, Cyclone~\cite{JimUSENIX02}, 
SVA~\cite{CriswellSOSP07}, etc.).
Static checking for memory safety is certainly possible 
when it is part of language design (for instance, using type systems as in Rust~\cite{Rust}).
Dynamic analysis techniques as in~\cite{valgrind2007,rosu-schulte-serbanuta-2009-rv,Serebryany2012} instrument program executables and report errors as they occur during program execution.
Recently, there has also been emerging interest in enforcing runtime memory safety using hardware and software support for tagged memory~\cite{Oleksenko2018,joannou2017,lowRISC,cheri2015,serebryany2018memory}.

This work stems from the recent decidability result on uninterpreted coherent programs~\cite{coherence2019}, which has also been extended to incorporate reasoning modulo theories including associativity and commutativity of functions over the data domain and ordering relations on the data domain~\cite{mathurtheories2019}. In our work we use automata-theoretic techniques that, over the data domain, reason about equality and function computation over the data elements. Further, for checking memory safety, our procedure tracks a subset of the allocated nodes, namely, the \emph{frontier} nodes. While a considerable portion of the literature on assertion checking and memory safety for heap-manipulating programs is devoted to techniques that compromise on either soundness, completeness, or decidability, there has been some work that aims at decidable reasoning, while still preserving some form of soundness and completeness.

The work in~\cite{Alur2011} reduces assertion checking of single-pass list-processing programs to questions on data string transducers that work over sequences of tagged data values. The decidablity is mainly a consequence of the fact that there is a \emph{single} variable that advances in a single-pass fashion. In contrast, our work defines a more general notion of single-pass programs using the \emph{\koherence} restriction that still allows for multiple variables to support pointer updates. Further, the work in~\cite{Alur2011} does not directly apply to verifying memory safety of programs as it does not explicitly handle freeing of memory locations, and the operations of the transducer cannot effect changes to the shape of the heap. Another key difference is that streaming data string transducers only allow reasoning about ordering and equality over data but cannot support more complex reasoning such as the congruence arising from function computation.

The work in~\cite{Bouajjani2006} proposes a class of list programs for which verification is decidable, and crucially relies on the idea of representing fragments of the allocated heap by a bounded number of segments and summaries about them, which is one among many other works~\cite{tvla,bardin2004,Balaban2005,Berdine2004,Manevich2005,Dor2000} that employ a similar approach. These works can handle limited reasoning with data such as total orders on the data domain, but again, do not support predicates like equality on data or function congruences resulting from equalities, and further, often address questions specific to lists.
Our work, on the other hand, tackles the more general problem of the verification of uninterpreted heap programs and instantiates the \alaw~ condition to the class of forest data-structures. 
The key idea of tracking the \emph{frontier} heap locations, which
we use for obtaining decidability of \koherent~ programs,
appears orthogonal to this line of work.

The work in~\cite{Bozga2007} differs fundamentally from ours in that pointer updates are forbidden, which is a salient feature of our work. Pointer updates are at the heart of the difficulty in building a theory of uninterpreted programs working over heaps. Additionally, that work
forbids nesting of loops as well as conditionals within loops, a restriction also used
 to obtain decidability in earlier work~\cite{godoy09} on uninterpreted programs;
there is no such syntactic restriction on the class of programs we introduce in this paper.

The work in~\cite{Bouajjani2005} over-approximates the set of heap configurations associated with a given program location as a regular language over a finite alphabet. The program transformations are represented by finite state transducers, and the work employs an abstraction refinement approach for verifying heap-manipulating programs. Since such abstraction refinement loops may not always terminate for arbitrary programs, the proposed approach is only a semi-decision procedure. Further, this work does not support reasoning over the data sort.
Other notable static analyses that employ abstraction refinement for verifying heap programs include the notable work on shape analysis~\cite{tvla, Yahav2001,lev-ami2000} and automatic predicate abstraction~\cite{Ball2001}.
Statically verifying memory safety using such
incomplete techniques can of course, and
commonly does, result in false positives.

There is a rich literature on program logics for
heap verification; in particular separation logics
~\cite{OHearn2001, Reynolds2002} and FO logics 
based on the principles of separation logic~\cite{LodingFrameLogicArxiv}. 
Decidable fragments of such logics have also been studied~\cite{Berdine2004,grasshopper2014,piskac2014cav,Piskac2013,Navarro2013,Navarro2011,berdine2006,Cook2011,Moller2001}.
However, typically, these decision procedures are for checking validity of Hoare triples, 
and the problem of generating loop invariants is often undecidable, as is the problem of completely automatic verification of programs against specifications expressed in these logics.
Some invariant generation techniques have been discovered for problems in this domain~\cite{infer,invsynthincomplete}, but are, of course, inherently incomplete.


\section{Conclusions and Future Work}
\seclabel{conclusion}
This paper establishes a foundational result for decidably verifying assertions in programs that update maps for a subclass that is alias-aware and coherent. We have used this general result to develop decidable verification of memory safety for a class of programs, called \koherent\ programs, working on forest data-structures.
We also proved membership of programs in this class is decidable. 
We showed through a prototype implementation of our decision procedure, 
and its evaluation on a set of single-pass algorithms, that forest data-structures typically fall in our decidable class, and that we can verify memory safety accurately for them.

The most compelling future direction is to adapt the technique in this paper to provide a memory safety analysis tool for a standard programming language (such as C/C++), handling the rest of the programming language using abstractions (e.g., arrays, allocation of varying blocks of memory, etc.). We believe that our automata-based algorithm will scale well. Realizing the techniques presented herein in a full-fledged memory safety analysis tool would be interesting.

On the theoretical front, there are several interesting directions.
First, we could ask how to generalize our results to go beyond \koherent~ programs on forest data-structures. 
Although forest data-structures are fairly common as initial heaps for many programs, finding a natural class of heap structures beyond forest 
data-structures where alias-awareness can be easily established seems
an interesting, challenging problem. We believe that data-structures such as doubly-linked lists and trees with parent pointers, and more generally, data-structures that have bounded tree-width with uniform tree decompositions may be amenable to our technique.
Tackling multi-pass algorithms on forest data-structures is also an interesting open direction.
We believe the best way to look at our work in a larger verification context is that single-pass \koherent\ programs are the new \emph{basic blocks} that can be completely automatically analyzed, despite the fact that they contain loops. 
Putting these blocks together to handle programs with multiple passes over data-structures, 
even in an incomplete fashion, is  an interesting future direction. 

Another dimension for exploration is 
to consider more general post-conditions that can be proved automatically and that go beyond simple assertions. 
One of the limitations of our work is that, though we have an implicit precondition that demands that data-structures are forests, we do not establish that the  data-structure in the post state
is also a forest. 
The ability to establish such a property will allow us to maintain the forest property of data-structures as an invariant across multiple calls from clients that manipulate a data-structure using a library of methods, by showing that each of the methods provably maintains this invariant.




\begin{acks}                            
We thank the anonymous reviewers of POPL for several comments that helped improve the paper.
Umang Mathur is partially supported by a Google PhD Fellowship. 
This material is based upon work supported by the National Science Foundation under Grants NSF CCF 1901069 and NSF CCF 1527395.
\end{acks}

\bibliography{references}


\begin{thebibliography}{54}


\ifx \showCODEN    \undefined \def \showCODEN     #1{\unskip}     \fi
\ifx \showDOI      \undefined \def \showDOI       #1{#1}\fi
\ifx \showISBNx    \undefined \def \showISBNx     #1{\unskip}     \fi
\ifx \showISBNxiii \undefined \def \showISBNxiii  #1{\unskip}     \fi
\ifx \showISSN     \undefined \def \showISSN      #1{\unskip}     \fi
\ifx \showLCCN     \undefined \def \showLCCN      #1{\unskip}     \fi
\ifx \shownote     \undefined \def \shownote      #1{#1}          \fi
\ifx \showarticletitle \undefined \def \showarticletitle #1{#1}   \fi
\ifx \showURL      \undefined \def \showURL       {\relax}        \fi
\providecommand\bibfield[2]{#2}
\providecommand\bibinfo[2]{#2}
\providecommand\natexlab[1]{#1}
\providecommand\showeprint[2][]{arXiv:#2}

\bibitem[\protect\citeauthoryear{Alur and \v{C}ern\'{y}}{Alur and
  \v{C}ern\'{y}}{2011}]%
        {Alur2011}
\bibfield{author}{\bibinfo{person}{Rajeev Alur} {and} \bibinfo{person}{Pavol
  \v{C}ern\'{y}}.} \bibinfo{year}{2011}\natexlab{}.
\newblock \showarticletitle{Streaming Transducers for Algorithmic Verification
  of Single-pass List-processing Programs}.
\newblock \bibinfo{journal}{\emph{SIGPLAN Not.}} \bibinfo{volume}{46},
  \bibinfo{number}{1} (\bibinfo{date}{Jan.} \bibinfo{year}{2011}),
  \bibinfo{pages}{599--610}.
\newblock
\showISSN{0362-1340}
\urldef\tempurl%
\url{https://doi.org/10.1145/1925844.1926454}
\showDOI{\tempurl}


\bibitem[\protect\citeauthoryear{Austin, Breach, and Sohi}{Austin
  et~al\mbox{.}}{1994}]%
        {AustinEtAlPLDI94}
\bibfield{author}{\bibinfo{person}{Todd~M. Austin}, \bibinfo{person}{Scott~E.
  Breach}, {and} \bibinfo{person}{Gurindar~S. Sohi}.}
  \bibinfo{year}{1994}\natexlab{}.
\newblock \showarticletitle{Efficient Detection of All Pointer and Array Access
  Errors}. In \bibinfo{booktitle}{\emph{Proceedings of the ACM SIGPLAN 1994
  Conference on Programming Language Design and Implementation}}
  \emph{(\bibinfo{series}{PLDI '94})}. \bibinfo{publisher}{ACM},
  \bibinfo{address}{New York, NY, USA}, \bibinfo{pages}{290--301}.
\newblock
\showISBNx{0-89791-662-X}
\urldef\tempurl%
\url{https://doi.org/10.1145/178243.178446}
\showDOI{\tempurl}


\bibitem[\protect\citeauthoryear{Balaban, Pnueli, and Zuck}{Balaban
  et~al\mbox{.}}{2005}]%
        {Balaban2005}
\bibfield{author}{\bibinfo{person}{Ittai Balaban}, \bibinfo{person}{Amir
  Pnueli}, {and} \bibinfo{person}{Lenore~D. Zuck}.}
  \bibinfo{year}{2005}\natexlab{}.
\newblock \showarticletitle{Shape Analysis by Predicate Abstraction}. In
  \bibinfo{booktitle}{\emph{Proceedings of the 6th International Conference on
  Verification, Model Checking, and Abstract Interpretation}}
  \emph{(\bibinfo{series}{VMCAI'05})}. \bibinfo{publisher}{Springer-Verlag},
  \bibinfo{address}{Berlin, Heidelberg}, \bibinfo{pages}{164--180}.
\newblock
\showISBNx{3-540-24297-X, 978-3-540-24297-0}
\urldef\tempurl%
\url{https://doi.org/10.1007/978-3-540-30579-8_12}
\showDOI{\tempurl}


\bibitem[\protect\citeauthoryear{Ball, Majumdar, Millstein, and Rajamani}{Ball
  et~al\mbox{.}}{2001}]%
        {Ball2001}
\bibfield{author}{\bibinfo{person}{Thomas Ball}, \bibinfo{person}{Rupak
  Majumdar}, \bibinfo{person}{Todd Millstein}, {and} \bibinfo{person}{Sriram~K.
  Rajamani}.} \bibinfo{year}{2001}\natexlab{}.
\newblock \showarticletitle{Automatic Predicate Abstraction of C Programs}. In
  \bibinfo{booktitle}{\emph{Proceedings of the ACM SIGPLAN 2001 Conference on
  Programming Language Design and Implementation}} \emph{(\bibinfo{series}{PLDI
  '01})}. \bibinfo{publisher}{ACM}, \bibinfo{address}{New York, NY, USA},
  \bibinfo{pages}{203--213}.
\newblock
\showISBNx{1-58113-414-2}
\urldef\tempurl%
\url{https://doi.org/10.1145/378795.378846}
\showDOI{\tempurl}


\bibitem[\protect\citeauthoryear{Bardin, Finkel, and Nowak}{Bardin
  et~al\mbox{.}}{2004}]%
        {bardin2004}
\bibfield{author}{\bibinfo{person}{S{\'e}bastien Bardin},
  \bibinfo{person}{Alain Finkel}, {and} \bibinfo{person}{David Nowak}.}
  \bibinfo{year}{2004}\natexlab{}.
\newblock \showarticletitle{Toward symbolic verification of programs handling
  pointers}.
\newblock  (\bibinfo{year}{2004}).
\newblock


\bibitem[\protect\citeauthoryear{Berdine, Calcagno, and O'Hearn}{Berdine
  et~al\mbox{.}}{2004}]%
        {Berdine2004}
\bibfield{author}{\bibinfo{person}{Josh Berdine}, \bibinfo{person}{Cristiano
  Calcagno}, {and} \bibinfo{person}{Peter~W. O'Hearn}.}
  \bibinfo{year}{2004}\natexlab{}.
\newblock \showarticletitle{A Decidable Fragment of Separation Logic}. In
  \bibinfo{booktitle}{\emph{Proceedings of the 24th International Conference on
  Foundations of Software Technology and Theoretical Computer Science}}
  \emph{(\bibinfo{series}{FSTTCS'04})}. \bibinfo{publisher}{Springer-Verlag},
  \bibinfo{address}{Berlin, Heidelberg}, \bibinfo{pages}{97--109}.
\newblock
\showISBNx{3-540-24058-6, 978-3-540-24058-7}
\urldef\tempurl%
\url{https://doi.org/10.1007/978-3-540-30538-5_9}
\showDOI{\tempurl}


\bibitem[\protect\citeauthoryear{Berdine, Calcagno, and O'Hearn}{Berdine
  et~al\mbox{.}}{2006}]%
        {berdine2006}
\bibfield{author}{\bibinfo{person}{Josh Berdine}, \bibinfo{person}{Cristiano
  Calcagno}, {and} \bibinfo{person}{Peter~W. O'Hearn}.}
  \bibinfo{year}{2006}\natexlab{}.
\newblock \showarticletitle{Smallfoot: Modular Automatic Assertion Checking
  with Separation Logic}. In \bibinfo{booktitle}{\emph{Formal Methods for
  Components and Objects}}, \bibfield{editor}{\bibinfo{person}{Frank~S.
  de~Boer}, \bibinfo{person}{Marcello~M. Bonsangue}, \bibinfo{person}{Susanne
  Graf}, {and} \bibinfo{person}{Willem-Paul de~Roever}} (Eds.).
  \bibinfo{publisher}{Springer Berlin Heidelberg}, \bibinfo{address}{Berlin,
  Heidelberg}, \bibinfo{pages}{115--137}.
\newblock


\bibitem[\protect\citeauthoryear{Bouajjani, Bozga, Habermehl, Iosif, Moro, and
  Vojnar}{Bouajjani et~al\mbox{.}}{2006}]%
        {Bouajjani2006}
\bibfield{author}{\bibinfo{person}{Ahmed Bouajjani}, \bibinfo{person}{Marius
  Bozga}, \bibinfo{person}{Peter Habermehl}, \bibinfo{person}{Radu Iosif},
  \bibinfo{person}{Pierre Moro}, {and} \bibinfo{person}{Tom{\'a}{\v{s}}
  Vojnar}.} \bibinfo{year}{2006}\natexlab{}.
\newblock \showarticletitle{Programs with Lists Are Counter Automata}. In
  \bibinfo{booktitle}{\emph{Computer Aided Verification}},
  \bibfield{editor}{\bibinfo{person}{Thomas Ball} {and}
  \bibinfo{person}{Robert~B. Jones}} (Eds.). \bibinfo{publisher}{Springer
  Berlin Heidelberg}, \bibinfo{address}{Berlin, Heidelberg},
  \bibinfo{pages}{517--531}.
\newblock


\bibitem[\protect\citeauthoryear{Bouajjani, Habermehl, Moro, and
  Vojnar}{Bouajjani et~al\mbox{.}}{2005}]%
        {Bouajjani2005}
\bibfield{author}{\bibinfo{person}{Ahmed Bouajjani}, \bibinfo{person}{Peter
  Habermehl}, \bibinfo{person}{Pierre Moro}, {and}
  \bibinfo{person}{Tom\'{a}\v{s} Vojnar}.} \bibinfo{year}{2005}\natexlab{}.
\newblock \showarticletitle{Verifying Programs with Dynamic 1-selector-linked
  Structures in Regular Model Checking}. In
  \bibinfo{booktitle}{\emph{Proceedings of the 11th International Conference on
  Tools and Algorithms for the Construction and Analysis of Systems}}
  \emph{(\bibinfo{series}{TACAS'05})}. \bibinfo{publisher}{Springer-Verlag},
  \bibinfo{address}{Berlin, Heidelberg}, \bibinfo{pages}{13--29}.
\newblock
\showISBNx{3-540-25333-5, 978-3-540-25333-4}
\urldef\tempurl%
\url{https://doi.org/10.1007/978-3-540-31980-1_2}
\showDOI{\tempurl}


\bibitem[\protect\citeauthoryear{Bozga and Iosif}{Bozga and Iosif}{2007}]%
        {Bozga2007}
\bibfield{author}{\bibinfo{person}{Marius Bozga} {and} \bibinfo{person}{Radu
  Iosif}.} \bibinfo{year}{2007}\natexlab{}.
\newblock \showarticletitle{On Flat Programs with Lists}. In
  \bibinfo{booktitle}{\emph{Proceedings of the 8th International Conference on
  Verification, Model Checking, and Abstract Interpretation}}
  \emph{(\bibinfo{series}{VMCAI'07})}. \bibinfo{publisher}{Springer-Verlag},
  \bibinfo{address}{Berlin, Heidelberg}, \bibinfo{pages}{122--136}.
\newblock
\showISBNx{978-3-540-69735-0}
\urldef\tempurl%
\url{http://dl.acm.org/citation.cfm?id=1763048.1763061}
\showURL{%
\tempurl}


\bibitem[\protect\citeauthoryear{Bradley and Manna}{Bradley and Manna}{2007}]%
        {calcofcomputation}
\bibfield{author}{\bibinfo{person}{Aaron~R. Bradley} {and}
  \bibinfo{person}{Zohar Manna}.} \bibinfo{year}{2007}\natexlab{}.
\newblock \bibinfo{booktitle}{\emph{The Calculus of Computation: Decision
  Procedures with Applications to Verification}}.
\newblock \bibinfo{publisher}{Springer-Verlag}, \bibinfo{address}{Berlin,
  Heidelberg}.
\newblock
\showISBNx{3540741127}


\bibitem[\protect\citeauthoryear{Calcagno, Distefano, O'Hearn, and
  Yang}{Calcagno et~al\mbox{.}}{2011}]%
        {infer}
\bibfield{author}{\bibinfo{person}{Cristiano Calcagno}, \bibinfo{person}{Dino
  Distefano}, \bibinfo{person}{Peter~W. O'Hearn}, {and}
  \bibinfo{person}{Hongseok Yang}.} \bibinfo{year}{2011}\natexlab{}.
\newblock \showarticletitle{Compositional Shape Analysis by Means of
  Bi-Abduction}.
\newblock \bibinfo{journal}{\emph{J. ACM}} \bibinfo{volume}{58},
  \bibinfo{number}{6}, Article \bibinfo{articleno}{26} (\bibinfo{date}{Dec.}
  \bibinfo{year}{2011}), \bibinfo{numpages}{66}~pages.
\newblock
\showISSN{0004-5411}
\urldef\tempurl%
\url{https://doi.org/10.1145/2049697.2049700}
\showDOI{\tempurl}


\bibitem[\protect\citeauthoryear{Condit, Harren, McPeak, Necula, and
  Weimer}{Condit et~al\mbox{.}}{2003}]%
        {ConditPLDI03}
\bibfield{author}{\bibinfo{person}{Jeremy Condit}, \bibinfo{person}{Matthew
  Harren}, \bibinfo{person}{Scott McPeak}, \bibinfo{person}{George~C. Necula},
  {and} \bibinfo{person}{Westley Weimer}.} \bibinfo{year}{2003}\natexlab{}.
\newblock \showarticletitle{CCured in the real world}. In
  \bibinfo{booktitle}{\emph{Proceedings of the {ACM} {SIGPLAN} 2003 Conference
  on Programming Language Design and Implementation 2003, San Diego,
  California, USA, June 9-11, 2003}}. \bibinfo{pages}{232--244}.
\newblock
\urldef\tempurl%
\url{https://doi.org/10.1145/781131.781157}
\showDOI{\tempurl}


\bibitem[\protect\citeauthoryear{Cook, Haase, Ouaknine, Parkinson, and
  Worrell}{Cook et~al\mbox{.}}{2011}]%
        {Cook2011}
\bibfield{author}{\bibinfo{person}{Byron Cook}, \bibinfo{person}{Christoph
  Haase}, \bibinfo{person}{Jo\"{e}l Ouaknine}, \bibinfo{person}{Matthew
  Parkinson}, {and} \bibinfo{person}{James Worrell}.}
  \bibinfo{year}{2011}\natexlab{}.
\newblock \showarticletitle{Tractable Reasoning in a Fragment of Separation
  Logic}. In \bibinfo{booktitle}{\emph{Proceedings of the 22Nd International
  Conference on Concurrency Theory}} \emph{(\bibinfo{series}{CONCUR'11})}.
  \bibinfo{publisher}{Springer-Verlag}, \bibinfo{address}{Berlin, Heidelberg},
  \bibinfo{pages}{235--249}.
\newblock
\showISBNx{978-3-642-23216-9}
\urldef\tempurl%
\url{http://dl.acm.org/citation.cfm?id=2040235.2040256}
\showURL{%
\tempurl}


\bibitem[\protect\citeauthoryear{Criswell, Lenharth, Dhurjati, and
  Adve}{Criswell et~al\mbox{.}}{2007}]%
        {CriswellSOSP07}
\bibfield{author}{\bibinfo{person}{John Criswell}, \bibinfo{person}{Andrew
  Lenharth}, \bibinfo{person}{Dinakar Dhurjati}, {and}
  \bibinfo{person}{Vikram~S. Adve}.} \bibinfo{year}{2007}\natexlab{}.
\newblock \showarticletitle{Secure virtual architecture: a safe execution
  environment for commodity operating systems}. In
  \bibinfo{booktitle}{\emph{Proceedings of the 21st {ACM} Symposium on
  Operating Systems Principles 2007, {SOSP} 2007, Stevenson, Washington, USA,
  October 14-17, 2007}}. \bibinfo{pages}{351--366}.
\newblock
\urldef\tempurl%
\url{https://doi.org/10.1145/1294261.1294295}
\showDOI{\tempurl}


\bibitem[\protect\citeauthoryear{Dor, Rodeh, and Sagiv}{Dor
  et~al\mbox{.}}{2000}]%
        {Dor2000}
\bibfield{author}{\bibinfo{person}{Nurit Dor}, \bibinfo{person}{Michael Rodeh},
  {and} \bibinfo{person}{Shmuel Sagiv}.} \bibinfo{year}{2000}\natexlab{}.
\newblock \showarticletitle{Checking Cleanness in Linked Lists}. In
  \bibinfo{booktitle}{\emph{SAS}}.
\newblock


\bibitem[\protect\citeauthoryear{Floyd}{Floyd}{1967}]%
        {FloydMeaning1967}
\bibfield{author}{\bibinfo{person}{Robert~W. Floyd}.}
  \bibinfo{year}{1967}\natexlab{}.
\newblock \showarticletitle{Assigning meanings to programs}.
\newblock \bibinfo{journal}{\emph{Mathematical aspects of computer science}}
  \bibinfo{volume}{19}, \bibinfo{number}{19-32} (\bibinfo{year}{1967}),
  \bibinfo{pages}{1}.
\newblock
\urldef\tempurl%
\url{http://www.cs.ucdavis.edu/~su/teaching/ecs240-s09/readings/FloydMeaning.pdf}
\showURL{%
\tempurl}


\bibitem[\protect\citeauthoryear{Godoy and Tiwari}{Godoy and Tiwari}{2009}]%
        {godoy09}
\bibfield{author}{\bibinfo{person}{Guillem Godoy} {and} \bibinfo{person}{Ashish
  Tiwari}.} \bibinfo{year}{2009}\natexlab{}.
\newblock \showarticletitle{Invariant Checking for Programs with Procedure
  Calls}. In \bibinfo{booktitle}{\emph{Proceedings of the 16th International
  Symposium on Static Analysis}} \emph{(\bibinfo{series}{SAS '09})}.
  \bibinfo{publisher}{Springer-Verlag}, \bibinfo{address}{Berlin, Heidelberg},
  \bibinfo{pages}{326--342}.
\newblock
\showISBNx{978-3-642-03236-3}


\bibitem[\protect\citeauthoryear{Hicks}{Hicks}{2014}]%
        {PLEnthusiast}
\bibfield{author}{\bibinfo{person}{Michael Hicks}.}
  \bibinfo{year}{2014}\natexlab{}.
\newblock \bibinfo{title}{{{What is memory safety?, The Programming Languages
  Enthusiast}}}.
\newblock
  \bibinfo{howpublished}{\url{http://www.pl-enthusiast.net/2014/07/21/memory-safety/}}.
\newblock
\newblock
\shownote{Accessed: 2019-04-05.}


\bibitem[\protect\citeauthoryear{Hoare}{Hoare}{1969}]%
        {Hoare1969}
\bibfield{author}{\bibinfo{person}{C.~A.~R. Hoare}.}
  \bibinfo{year}{1969}\natexlab{}.
\newblock \showarticletitle{An Axiomatic Basis for Computer Programming}.
\newblock \bibinfo{journal}{\emph{Commun. ACM}} \bibinfo{volume}{12},
  \bibinfo{number}{10} (\bibinfo{date}{Oct.} \bibinfo{year}{1969}),
  \bibinfo{pages}{576--580}.
\newblock
\showISSN{0001-0782}
\urldef\tempurl%
\url{https://doi.org/10.1145/363235.363259}
\showDOI{\tempurl}


\bibitem[\protect\citeauthoryear{Jim, Morrisett, Grossman, Hicks, Cheney, and
  Wang}{Jim et~al\mbox{.}}{2002}]%
        {JimUSENIX02}
\bibfield{author}{\bibinfo{person}{Trevor Jim}, \bibinfo{person}{J~Greg
  Morrisett}, \bibinfo{person}{Dan Grossman}, \bibinfo{person}{Michael~W
  Hicks}, \bibinfo{person}{James Cheney}, {and} \bibinfo{person}{Yanling
  Wang}.} \bibinfo{year}{2002}\natexlab{}.
\newblock \showarticletitle{Cyclone: A Safe Dialect of C}. In
  \bibinfo{booktitle}{\emph{Proceedings of the General Track of the annual
  conference on USENIX Annual Technical Conference}}. USENIX Association,
  \bibinfo{pages}{275--288}.
\newblock


\bibitem[\protect\citeauthoryear{{Joannou}, {Woodruff}, {Kovacsics}, {Moore},
  {Bradbury}, {Xia}, {Watson}, {Chisnall}, {Roe}, {Davis}, {Napierala},
  {Baldwin}, {Gudka}, {Neumann}, {Mazzinghi}, {Richardson}, {Son}, and
  {Markettos}}{{Joannou} et~al\mbox{.}}{2017}]%
        {joannou2017}
\bibfield{author}{\bibinfo{person}{A. {Joannou}}, \bibinfo{person}{J.
  {Woodruff}}, \bibinfo{person}{R. {Kovacsics}}, \bibinfo{person}{S.~W.
  {Moore}}, \bibinfo{person}{A. {Bradbury}}, \bibinfo{person}{H. {Xia}},
  \bibinfo{person}{R.~N.~M. {Watson}}, \bibinfo{person}{D. {Chisnall}},
  \bibinfo{person}{M. {Roe}}, \bibinfo{person}{B. {Davis}}, \bibinfo{person}{E.
  {Napierala}}, \bibinfo{person}{J. {Baldwin}}, \bibinfo{person}{K. {Gudka}},
  \bibinfo{person}{P.~G. {Neumann}}, \bibinfo{person}{A. {Mazzinghi}},
  \bibinfo{person}{A. {Richardson}}, \bibinfo{person}{S. {Son}}, {and}
  \bibinfo{person}{A.~T. {Markettos}}.} \bibinfo{year}{2017}\natexlab{}.
\newblock \showarticletitle{Efficient Tagged Memory}. In
  \bibinfo{booktitle}{\emph{2017 IEEE International Conference on Computer
  Design (ICCD)}}. \bibinfo{pages}{641--648}.
\newblock
\urldef\tempurl%
\url{https://doi.org/10.1109/ICCD.2017.112}
\showDOI{\tempurl}


\bibitem[\protect\citeauthoryear{Lev-Ami and Sagiv}{Lev-Ami and Sagiv}{2000}]%
        {lev-ami2000}
\bibfield{author}{\bibinfo{person}{Tal Lev-Ami} {and} \bibinfo{person}{Mooly
  Sagiv}.} \bibinfo{year}{2000}\natexlab{}.
\newblock \showarticletitle{TVLA: A System for Implementing Static Analyses}.
  In \bibinfo{booktitle}{\emph{Static Analysis}},
  \bibfield{editor}{\bibinfo{person}{Jens Palsberg}} (Ed.).
  \bibinfo{publisher}{Springer Berlin Heidelberg}, \bibinfo{address}{Berlin,
  Heidelberg}, \bibinfo{pages}{280--301}.
\newblock


\bibitem[\protect\citeauthoryear{L{\o}ding, Madhusudan, Murali, and
  Pe}{L{\o}ding et~al\mbox{.}}{2019}]%
        {LodingFrameLogicArxiv}
\bibfield{author}{\bibinfo{person}{Christof L{\o}ding}, \bibinfo{person}{P.
  Madhusudan}, \bibinfo{person}{Adithya Murali}, {and} \bibinfo{person}{Lucas
  Pe}.} \bibinfo{year}{2019}\natexlab{}.
\newblock \showarticletitle{A First Order Logic with Frames}.
\newblock \bibinfo{journal}{\emph{CoRR}}  \bibinfo{volume}{abs/1910.09089}
  (\bibinfo{year}{2019}).
\newblock
\showeprint[arxiv]{1910.09089}
\urldef\tempurl%
\url{https://arxiv.org/abs/1901.09089}
\showURL{%
\tempurl}


\bibitem[\protect\citeauthoryear{lowRISC}{lowRISC}{2019}]%
        {lowRISC}
\bibfield{author}{\bibinfo{person}{lowRISC}.} \bibinfo{year}{2019}\natexlab{}.
\newblock \bibinfo{title}{{{lowRISC: A fully open-sourced, linux-capable,
  system-on-a-chip}}}.
\newblock \bibinfo{howpublished}{\url{https://www.lowrisc.org/.}}.
\newblock
\newblock
\shownote{Accessed: 2019-10-29.}


\bibitem[\protect\citeauthoryear{Manevich, Yahav, Ramalingam, and
  Sagiv}{Manevich et~al\mbox{.}}{2005}]%
        {Manevich2005}
\bibfield{author}{\bibinfo{person}{Roman Manevich}, \bibinfo{person}{E. Yahav},
  \bibinfo{person}{G. Ramalingam}, {and} \bibinfo{person}{Mooly Sagiv}.}
  \bibinfo{year}{2005}\natexlab{}.
\newblock \showarticletitle{Predicate Abstraction and Canonical Abstraction for
  Singly-Linked Lists}. In \bibinfo{booktitle}{\emph{Verification, Model
  Checking, and Abstract Interpretation}},
  \bibfield{editor}{\bibinfo{person}{Radhia Cousot}} (Ed.).
  \bibinfo{publisher}{Springer Berlin Heidelberg}, \bibinfo{address}{Berlin,
  Heidelberg}, \bibinfo{pages}{181--198}.
\newblock


\bibitem[\protect\citeauthoryear{Mathur, Madhusudan, and Viswanathan}{Mathur
  et~al\mbox{.}}{2019a}]%
        {coherence2019}
\bibfield{author}{\bibinfo{person}{Umang Mathur}, \bibinfo{person}{P.
  Madhusudan}, {and} \bibinfo{person}{Mahesh Viswanathan}.}
  \bibinfo{year}{2019}\natexlab{a}.
\newblock \showarticletitle{Decidable Verification of Uninterpreted Programs}.
\newblock \bibinfo{journal}{\emph{Proc. ACM Program. Lang.}}
  \bibinfo{volume}{3}, \bibinfo{number}{POPL}, Article \bibinfo{articleno}{46}
  (\bibinfo{date}{Jan.} \bibinfo{year}{2019}), \bibinfo{numpages}{29}~pages.
\newblock
\showISSN{2475-1421}
\urldef\tempurl%
\url{https://doi.org/10.1145/3290359}
\showDOI{\tempurl}


\bibitem[\protect\citeauthoryear{Mathur, Madhusudan, and Viswanathan}{Mathur
  et~al\mbox{.}}{2019b}]%
        {mathurtheories2019}
\bibfield{author}{\bibinfo{person}{Umang Mathur}, \bibinfo{person}{P.
  Madhusudan}, {and} \bibinfo{person}{Mahesh Viswanathan}.}
  \bibinfo{year}{2019}\natexlab{b}.
\newblock \showarticletitle{What's Decidable About Program Verification Modulo
  Axioms?}
\newblock \bibinfo{journal}{\emph{CoRR}}  \bibinfo{volume}{abs/1910.10889}
  (\bibinfo{year}{2019}).
\newblock
\showeprint[arxiv]{1910.10889}
\urldef\tempurl%
\url{https://arxiv.org/abs/1910.10889}
\showURL{%
\tempurl}


\bibitem[\protect\citeauthoryear{Mathur, Murali, Krogmeier, Madhusudan, and
  Viswanathan}{Mathur et~al\mbox{.}}{2019c}]%
        {streamveriftool}
\bibfield{author}{\bibinfo{person}{Umang Mathur}, \bibinfo{person}{Adithya
  Murali}, \bibinfo{person}{Paul Krogmeier}, \bibinfo{person}{P. Madhusudan},
  {and} \bibinfo{person}{Mahesh Viswanathan}.}
  \bibinfo{year}{2019}\natexlab{c}.
\newblock \bibinfo{title}{{{StreamVerif : Automata Based Verification of
  Uninterpreted Programs}}}.
\newblock \bibinfo{howpublished}{\url{https://github.com/umangm/streamverif}}.
\newblock


\bibitem[\protect\citeauthoryear{Matsakis and Klock}{Matsakis and
  Klock}{2014}]%
        {Rust}
\bibfield{author}{\bibinfo{person}{Nicholas~D. Matsakis} {and}
  \bibinfo{person}{Felix~S. Klock, II}.} \bibinfo{year}{2014}\natexlab{}.
\newblock \showarticletitle{The Rust Language}. In
  \bibinfo{booktitle}{\emph{Proceedings of the 2014 ACM SIGAda Annual
  Conference on High Integrity Language Technology}}
  \emph{(\bibinfo{series}{HILT '14})}. \bibinfo{publisher}{ACM},
  \bibinfo{address}{New York, NY, USA}, \bibinfo{pages}{103--104}.
\newblock
\showISBNx{978-1-4503-3217-0}
\urldef\tempurl%
\url{https://doi.org/10.1145/2663171.2663188}
\showDOI{\tempurl}


\bibitem[\protect\citeauthoryear{Microsoft}{Microsoft}{2019}]%
        {slashdotpost}
\bibfield{author}{\bibinfo{person}{Microsoft}.}
  \bibinfo{year}{2019}\natexlab{}.
\newblock \bibinfo{title}{{{70 Percent of All Security Bugs Are Memory Safety
  Issues}}}.
\newblock
  \bibinfo{howpublished}{\url{https://it.slashdot.org/story/19/02/11/2019247/microsoft-70-percent-of-all-security-bugs-are-memory-safety-issues}}.
\newblock
\newblock
\shownote{Accessed: 2019-04-05.}


\bibitem[\protect\citeauthoryear{M{\o}ller and Schwartzbach}{M{\o}ller and
  Schwartzbach}{2001}]%
        {Moller2001}
\bibfield{author}{\bibinfo{person}{Anders M{\o}ller} {and}
  \bibinfo{person}{Michael~I. Schwartzbach}.} \bibinfo{year}{2001}\natexlab{}.
\newblock \showarticletitle{The Pointer Assertion Logic Engine}. In
  \bibinfo{booktitle}{\emph{Proceedings of the ACM SIGPLAN 2001 Conference on
  Programming Language Design and Implementation}} \emph{(\bibinfo{series}{PLDI
  '01})}. \bibinfo{publisher}{ACM}, \bibinfo{address}{New York, NY, USA},
  \bibinfo{pages}{221--231}.
\newblock
\showISBNx{1-58113-414-2}
\urldef\tempurl%
\url{https://doi.org/10.1145/378795.378851}
\showDOI{\tempurl}


\bibitem[\protect\citeauthoryear{Nagarakatte, Martin, and
  Zdancewic}{Nagarakatte et~al\mbox{.}}{2015}]%
        {nagarakatte2015}
\bibfield{author}{\bibinfo{person}{Santosh Nagarakatte}, \bibinfo{person}{Milo
  M.~K. Martin}, {and} \bibinfo{person}{Steve Zdancewic}.}
  \bibinfo{year}{2015}\natexlab{}.
\newblock \showarticletitle{{Everything You Want to Know About Pointer-Based
  Checking}}. In \bibinfo{booktitle}{\emph{1st Summit on Advances in
  Programming Languages (SNAPL 2015)}} \emph{(\bibinfo{series}{Leibniz
  International Proceedings in Informatics (LIPIcs)})},
  \bibfield{editor}{\bibinfo{person}{Thomas Ball}, \bibinfo{person}{Rastislav
  Bodik}, \bibinfo{person}{Shriram Krishnamurthi}, \bibinfo{person}{Benjamin~S.
  Lerner}, {and} \bibinfo{person}{Greg Morrisett}} (Eds.),
  Vol.~\bibinfo{volume}{32}. \bibinfo{publisher}{Schloss
  Dagstuhl--Leibniz-Zentrum fuer Informatik}, \bibinfo{address}{Dagstuhl,
  Germany}, \bibinfo{pages}{190--208}.
\newblock
\showISBNx{978-3-939897-80-4}
\showISSN{1868-8969}
\urldef\tempurl%
\url{https://doi.org/10.4230/LIPIcs.SNAPL.2015.190}
\showDOI{\tempurl}


\bibitem[\protect\citeauthoryear{Navarro~P{\'e}rez and
  Rybalchenko}{Navarro~P{\'e}rez and Rybalchenko}{2011}]%
        {Navarro2011}
\bibfield{author}{\bibinfo{person}{Juan~Antonio Navarro~P{\'e}rez} {and}
  \bibinfo{person}{Andrey Rybalchenko}.} \bibinfo{year}{2011}\natexlab{}.
\newblock \showarticletitle{Separation Logic + Superposition Calculus = Heap
  Theorem Prover}. In \bibinfo{booktitle}{\emph{Proceedings of the 32Nd ACM
  SIGPLAN Conference on Programming Language Design and Implementation}}
  \emph{(\bibinfo{series}{PLDI '11})}. \bibinfo{publisher}{ACM},
  \bibinfo{address}{New York, NY, USA}, \bibinfo{pages}{556--566}.
\newblock
\showISBNx{978-1-4503-0663-8}
\urldef\tempurl%
\url{https://doi.org/10.1145/1993498.1993563}
\showDOI{\tempurl}


\bibitem[\protect\citeauthoryear{Navarro~P{\'e}rez and
  Rybalchenko}{Navarro~P{\'e}rez and Rybalchenko}{2013}]%
        {Navarro2013}
\bibfield{author}{\bibinfo{person}{Juan~Antonio Navarro~P{\'e}rez} {and}
  \bibinfo{person}{Andrey Rybalchenko}.} \bibinfo{year}{2013}\natexlab{}.
\newblock \showarticletitle{Separation Logic Modulo Theories}. In
  \bibinfo{booktitle}{\emph{Programming Languages and Systems}},
  \bibfield{editor}{\bibinfo{person}{Chung-chieh Shan}} (Ed.).
  \bibinfo{publisher}{Springer International Publishing},
  \bibinfo{address}{Cham}, \bibinfo{pages}{90--106}.
\newblock


\bibitem[\protect\citeauthoryear{Necula, Condit, Harren, McPeak, and
  Weimer}{Necula et~al\mbox{.}}{2005}]%
        {NeculaTOPLAS05}
\bibfield{author}{\bibinfo{person}{George~C. Necula}, \bibinfo{person}{Jeremy
  Condit}, \bibinfo{person}{Matthew Harren}, \bibinfo{person}{Scott McPeak},
  {and} \bibinfo{person}{Westley Weimer}.} \bibinfo{year}{2005}\natexlab{}.
\newblock \showarticletitle{CCured: type-safe retrofitting of legacy software}.
\newblock \bibinfo{journal}{\emph{{ACM} Trans. Program. Lang. Syst.}}
  \bibinfo{volume}{27}, \bibinfo{number}{3} (\bibinfo{year}{2005}),
  \bibinfo{pages}{477--526}.
\newblock
\urldef\tempurl%
\url{https://doi.org/10.1145/1065887.1065892}
\showDOI{\tempurl}


\bibitem[\protect\citeauthoryear{Necula, McPeak, and Weimer}{Necula
  et~al\mbox{.}}{2002}]%
        {necula2002ccured}
\bibfield{author}{\bibinfo{person}{George~C Necula}, \bibinfo{person}{Scott
  McPeak}, {and} \bibinfo{person}{Westley Weimer}.}
  \bibinfo{year}{2002}\natexlab{}.
\newblock \showarticletitle{CCured: Type-safe retrofitting of legacy code}. In
  \bibinfo{booktitle}{\emph{ACM SIGPLAN Notices}}, Vol.~\bibinfo{volume}{37}.
  ACM, \bibinfo{pages}{128--139}.
\newblock


\bibitem[\protect\citeauthoryear{Neider, Garg, Madhusudan, Saha, and
  Park}{Neider et~al\mbox{.}}{2018}]%
        {invsynthincomplete}
\bibfield{author}{\bibinfo{person}{Daniel Neider}, \bibinfo{person}{Pranav
  Garg}, \bibinfo{person}{P. Madhusudan}, \bibinfo{person}{Shambwaditya Saha},
  {and} \bibinfo{person}{Daejun Park}.} \bibinfo{year}{2018}\natexlab{}.
\newblock \showarticletitle{Invariant Synthesis for Incomplete Verification
  Engines}. In \bibinfo{booktitle}{\emph{Tools and Algorithms for the
  Construction and Analysis of Systems}},
  \bibfield{editor}{\bibinfo{person}{Dirk Beyer} {and} \bibinfo{person}{Marieke
  Huisman}} (Eds.). \bibinfo{publisher}{Springer International Publishing},
  \bibinfo{address}{Cham}, \bibinfo{pages}{232--250}.
\newblock
\showISBNx{978-3-319-89960-2}


\bibitem[\protect\citeauthoryear{Nethercote and Seward}{Nethercote and
  Seward}{2007}]%
        {valgrind2007}
\bibfield{author}{\bibinfo{person}{Nicholas Nethercote} {and}
  \bibinfo{person}{Julian Seward}.} \bibinfo{year}{2007}\natexlab{}.
\newblock \showarticletitle{Valgrind: A Framework for Heavyweight Dynamic
  Binary Instrumentation}. In \bibinfo{booktitle}{\emph{Proceedings of the 28th
  ACM SIGPLAN Conference on Programming Language Design and Implementation}}
  \emph{(\bibinfo{series}{PLDI '07})}. \bibinfo{publisher}{ACM},
  \bibinfo{address}{New York, NY, USA}, \bibinfo{pages}{89--100}.
\newblock
\showISBNx{978-1-59593-633-2}
\urldef\tempurl%
\url{https://doi.org/10.1145/1250734.1250746}
\showDOI{\tempurl}


\bibitem[\protect\citeauthoryear{O'Hearn, Reynolds, and Yang}{O'Hearn
  et~al\mbox{.}}{2001}]%
        {OHearn2001}
\bibfield{author}{\bibinfo{person}{Peter~W. O'Hearn}, \bibinfo{person}{John~C.
  Reynolds}, {and} \bibinfo{person}{Hongseok Yang}.}
  \bibinfo{year}{2001}\natexlab{}.
\newblock \showarticletitle{Local Reasoning About Programs That Alter Data
  Structures}. In \bibinfo{booktitle}{\emph{Proceedings of the 15th
  International Workshop on Computer Science Logic}}
  \emph{(\bibinfo{series}{CSL '01})}. \bibinfo{publisher}{Springer-Verlag},
  \bibinfo{address}{London, UK, UK}, \bibinfo{pages}{1--19}.
\newblock
\showISBNx{3-540-42554-3}
\urldef\tempurl%
\url{http://dl.acm.org/citation.cfm?id=647851.737404}
\showURL{%
\tempurl}


\bibitem[\protect\citeauthoryear{Oleksenko, Kuvaiskii, Bhatotia, Felber, and
  Fetzer}{Oleksenko et~al\mbox{.}}{2018}]%
        {Oleksenko2018}
\bibfield{author}{\bibinfo{person}{Oleksii Oleksenko}, \bibinfo{person}{Dmitrii
  Kuvaiskii}, \bibinfo{person}{Pramod Bhatotia}, \bibinfo{person}{Pascal
  Felber}, {and} \bibinfo{person}{Christof Fetzer}.}
  \bibinfo{year}{2018}\natexlab{}.
\newblock \showarticletitle{Intel MPX Explained: A Cross-layer Analysis of the
  Intel MPX System Stack}.
\newblock \bibinfo{journal}{\emph{Proc. ACM Meas. Anal. Comput. Syst.}}
  \bibinfo{volume}{2}, \bibinfo{number}{2}, Article \bibinfo{articleno}{28}
  (\bibinfo{date}{June} \bibinfo{year}{2018}), \bibinfo{numpages}{30}~pages.
\newblock
\showISSN{2476-1249}
\urldef\tempurl%
\url{https://doi.org/10.1145/3224423}
\showDOI{\tempurl}


\bibitem[\protect\citeauthoryear{Piskac, Wies, and Zufferey}{Piskac
  et~al\mbox{.}}{2013}]%
        {Piskac2013}
\bibfield{author}{\bibinfo{person}{Ruzica Piskac}, \bibinfo{person}{Thomas
  Wies}, {and} \bibinfo{person}{Damien Zufferey}.}
  \bibinfo{year}{2013}\natexlab{}.
\newblock \showarticletitle{Automating Separation Logic Using SMT}. In
  \bibinfo{booktitle}{\emph{Proceedings of the 25th International Conference on
  Computer Aided Verification - Volume 8044}} \emph{(\bibinfo{series}{CAV
  2013})}. \bibinfo{publisher}{Springer-Verlag New York, Inc.},
  \bibinfo{address}{New York, NY, USA}, \bibinfo{pages}{773--789}.
\newblock
\showISBNx{978-3-642-39798-1}
\urldef\tempurl%
\url{https://doi.org/10.1007/978-3-642-39799-8_54}
\showDOI{\tempurl}


\bibitem[\protect\citeauthoryear{Piskac, Wies, and Zufferey}{Piskac
  et~al\mbox{.}}{2014a}]%
        {piskac2014cav}
\bibfield{author}{\bibinfo{person}{Ruzica Piskac}, \bibinfo{person}{Thomas
  Wies}, {and} \bibinfo{person}{Damien Zufferey}.}
  \bibinfo{year}{2014}\natexlab{a}.
\newblock \showarticletitle{Automating Separation Logic with Trees and Data}.
  In \bibinfo{booktitle}{\emph{Computer Aided Verification}},
  \bibfield{editor}{\bibinfo{person}{Armin Biere} {and}
  \bibinfo{person}{Roderick Bloem}} (Eds.). \bibinfo{publisher}{Springer
  International Publishing}, \bibinfo{address}{Cham},
  \bibinfo{pages}{711--728}.
\newblock


\bibitem[\protect\citeauthoryear{Piskac, Wies, and Zufferey}{Piskac
  et~al\mbox{.}}{2014b}]%
        {grasshopper2014}
\bibfield{author}{\bibinfo{person}{Ruzica Piskac}, \bibinfo{person}{Thomas
  Wies}, {and} \bibinfo{person}{Damien Zufferey}.}
  \bibinfo{year}{2014}\natexlab{b}.
\newblock \showarticletitle{GRASShopper}. In \bibinfo{booktitle}{\emph{Tools
  and Algorithms for the Construction and Analysis of Systems}},
  \bibfield{editor}{\bibinfo{person}{Erika {\'A}brah{\'a}m} {and}
  \bibinfo{person}{Klaus Havelund}} (Eds.). \bibinfo{publisher}{Springer Berlin
  Heidelberg}, \bibinfo{address}{Berlin, Heidelberg},
  \bibinfo{pages}{124--139}.
\newblock
\showISBNx{978-3-642-54862-8}


\bibitem[\protect\citeauthoryear{Reynolds}{Reynolds}{2002}]%
        {Reynolds2002}
\bibfield{author}{\bibinfo{person}{John~C. Reynolds}.}
  \bibinfo{year}{2002}\natexlab{}.
\newblock \showarticletitle{Separation Logic: A Logic for Shared Mutable Data
  Structures}. In \bibinfo{booktitle}{\emph{Proceedings of the 17th Annual IEEE
  Symposium on Logic in Computer Science}} \emph{(\bibinfo{series}{LICS '02})}.
  \bibinfo{publisher}{IEEE Computer Society}, \bibinfo{address}{Washington, DC,
  USA}, \bibinfo{pages}{55--74}.
\newblock
\showISBNx{0-7695-1483-9}
\urldef\tempurl%
\url{http://dl.acm.org/citation.cfm?id=645683.664578}
\showURL{%
\tempurl}


\bibitem[\protect\citeauthoryear{Rosu, Schulte, and Serbanuta}{Rosu
  et~al\mbox{.}}{2009}]%
        {rosu-schulte-serbanuta-2009-rv}
\bibfield{author}{\bibinfo{person}{Grigore Rosu}, \bibinfo{person}{Wolfram
  Schulte}, {and} \bibinfo{person}{Traian~Florin Serbanuta}.}
  \bibinfo{year}{2009}\natexlab{}.
\newblock \showarticletitle{Runtime Verification of {C} Memory Safety}. In
  \bibinfo{booktitle}{\emph{Runtime Verification (RV'09)}}
  \emph{(\bibinfo{series}{Lecture Notes in Computer Science})},
  \bibfield{editor}{\bibinfo{person}{Saddek Bensalem} {and}
  \bibinfo{person}{Doron~A. Peled}} (Eds.), Vol.~\bibinfo{volume}{5779}.
  \bibinfo{pages}{132--152}.
\newblock


\bibitem[\protect\citeauthoryear{Safe-C}{Safe-C}{2019}]%
        {safeclib}
\bibfield{author}{\bibinfo{person}{Safe-C}.} \bibinfo{year}{2019}\natexlab{}.
\newblock \bibinfo{title}{{{Safe C Library}}}.
\newblock
  \bibinfo{howpublished}{\url{https://rurban.github.io/safeclib/doc/safec-3.4/index.html}}.
\newblock
\newblock
\shownote{Accessed: 2019-04-05.}


\bibitem[\protect\citeauthoryear{Sagiv, Reps, and Wilhelm}{Sagiv
  et~al\mbox{.}}{1999}]%
        {tvla}
\bibfield{author}{\bibinfo{person}{Mooly Sagiv}, \bibinfo{person}{Thomas Reps},
  {and} \bibinfo{person}{Reinhard Wilhelm}.} \bibinfo{year}{1999}\natexlab{}.
\newblock \showarticletitle{Parametric Shape Analysis via 3-valued Logic}. In
  \bibinfo{booktitle}{\emph{Proceedings of the 26th ACM SIGPLAN-SIGACT
  Symposium on Principles of Programming Languages}}
  \emph{(\bibinfo{series}{POPL '99})}. \bibinfo{publisher}{ACM},
  \bibinfo{address}{New York, NY, USA}, \bibinfo{pages}{105--118}.
\newblock
\showISBNx{1-58113-095-3}
\urldef\tempurl%
\url{https://doi.org/10.1145/292540.292552}
\showDOI{\tempurl}


\bibitem[\protect\citeauthoryear{Serebryany, Bruening, Potapenko, and
  Vyukov}{Serebryany et~al\mbox{.}}{2012}]%
        {Serebryany2012}
\bibfield{author}{\bibinfo{person}{Konstantin Serebryany},
  \bibinfo{person}{Derek Bruening}, \bibinfo{person}{Alexander Potapenko},
  {and} \bibinfo{person}{Dmitry Vyukov}.} \bibinfo{year}{2012}\natexlab{}.
\newblock \showarticletitle{AddressSanitizer: A Fast Address Sanity Checker}.
  In \bibinfo{booktitle}{\emph{Proceedings of the 2012 USENIX Conference on
  Annual Technical Conference}} \emph{(\bibinfo{series}{USENIX ATC'12})}.
  \bibinfo{publisher}{USENIX Association}, \bibinfo{address}{Berkeley, CA,
  USA}, \bibinfo{pages}{28--28}.
\newblock
\urldef\tempurl%
\url{http://dl.acm.org/citation.cfm?id=2342821.2342849}
\showURL{%
\tempurl}


\bibitem[\protect\citeauthoryear{Serebryany, Stepanov, Shlyapnikov,
  Tsyrklevich, and Vyukov}{Serebryany et~al\mbox{.}}{2018}]%
        {serebryany2018memory}
\bibfield{author}{\bibinfo{person}{Kostya Serebryany}, \bibinfo{person}{Evgenii
  Stepanov}, \bibinfo{person}{Aleksey Shlyapnikov}, \bibinfo{person}{Vlad
  Tsyrklevich}, {and} \bibinfo{person}{Dmitry Vyukov}.}
  \bibinfo{year}{2018}\natexlab{}.
\newblock \showarticletitle{Memory Tagging and how it improves C/C++ memory
  safety}.
\newblock \bibinfo{journal}{\emph{CoRR}}  \bibinfo{volume}{abs/1802.09517}
  (\bibinfo{year}{2018}).
\newblock
\showeprint[arxiv]{1802.09517}
\urldef\tempurl%
\url{http://arxiv.org/abs/1802.09517}
\showURL{%
\tempurl}


\bibitem[\protect\citeauthoryear{Szekeres, Payer, Wei, and Song}{Szekeres
  et~al\mbox{.}}{2013}]%
        {Szekeres2013}
\bibfield{author}{\bibinfo{person}{Laszlo Szekeres}, \bibinfo{person}{Mathias
  Payer}, \bibinfo{person}{Tao Wei}, {and} \bibinfo{person}{Dawn Song}.}
  \bibinfo{year}{2013}\natexlab{}.
\newblock \showarticletitle{SoK: Eternal War in Memory}. In
  \bibinfo{booktitle}{\emph{Proceedings of the 2013 IEEE Symposium on Security
  and Privacy}} \emph{(\bibinfo{series}{SP '13})}. \bibinfo{publisher}{IEEE
  Computer Society}, \bibinfo{address}{Washington, DC, USA},
  \bibinfo{pages}{48--62}.
\newblock
\showISBNx{978-0-7695-4977-4}
\urldef\tempurl%
\url{https://doi.org/10.1109/SP.2013.13}
\showDOI{\tempurl}


\bibitem[\protect\citeauthoryear{{The\ Rust\ Team}}{{The\ Rust\ Team}}{2019}]%
        {rustlangdoc}
\bibfield{author}{\bibinfo{person}{{The\ Rust\ Team}}.}
  \bibinfo{year}{2019}\natexlab{}.
\newblock \bibinfo{title}{The Rust programming language}.
\newblock
\newblock
\urldef\tempurl%
\url{https://www.rust-lang.org/}
\showURL{%
\tempurl}


\bibitem[\protect\citeauthoryear{Watson, Woodruff, Neumann, Moore, Anderson,
  Chisnall, Dave, Davis, Gudka, Laurie, Murdoch, Norton, Roe, Son, and
  Vadera}{Watson et~al\mbox{.}}{2015}]%
        {cheri2015}
\bibfield{author}{\bibinfo{person}{Robert N.~M. Watson},
  \bibinfo{person}{Jonathan Woodruff}, \bibinfo{person}{Peter~G. Neumann},
  \bibinfo{person}{Simon~W. Moore}, \bibinfo{person}{Jonathan Anderson},
  \bibinfo{person}{David Chisnall}, \bibinfo{person}{Nirav Dave},
  \bibinfo{person}{Brooks Davis}, \bibinfo{person}{Khilan Gudka},
  \bibinfo{person}{Ben Laurie}, \bibinfo{person}{Steven~J. Murdoch},
  \bibinfo{person}{Robert Norton}, \bibinfo{person}{Michael Roe},
  \bibinfo{person}{Stacey Son}, {and} \bibinfo{person}{Munraj Vadera}.}
  \bibinfo{year}{2015}\natexlab{}.
\newblock \showarticletitle{CHERI: A Hybrid Capability-System Architecture for
  Scalable Software Compartmentalization}. In
  \bibinfo{booktitle}{\emph{Proceedings of the 2015 IEEE Symposium on Security
  and Privacy}} \emph{(\bibinfo{series}{SP '15})}. \bibinfo{publisher}{IEEE
  Computer Society}, \bibinfo{address}{Washington, DC, USA},
  \bibinfo{pages}{20--37}.
\newblock
\showISBNx{978-1-4673-6949-7}
\urldef\tempurl%
\url{https://doi.org/10.1109/SP.2015.9}
\showDOI{\tempurl}


\bibitem[\protect\citeauthoryear{Yahav}{Yahav}{2001}]%
        {Yahav2001}
\bibfield{author}{\bibinfo{person}{Eran Yahav}.}
  \bibinfo{year}{2001}\natexlab{}.
\newblock \showarticletitle{Verifying Safety Properties of Concurrent Java
  Programs Using 3-valued Logic}. In \bibinfo{booktitle}{\emph{Proceedings of
  the 28th ACM SIGPLAN-SIGACT Symposium on Principles of Programming
  Languages}} \emph{(\bibinfo{series}{POPL '01})}. \bibinfo{publisher}{ACM},
  \bibinfo{address}{New York, NY, USA}, \bibinfo{pages}{27--40}.
\newblock
\showISBNx{1-58113-336-7}
\urldef\tempurl%
\url{https://doi.org/10.1145/360204.360206}
\showDOI{\tempurl}


\end{thebibliography}



\end{document}